\documentclass[]{gradient} 
\usepackage{array}

\usepackage{microtype}
\usepackage{makecell}
\usepackage{textgreek}
\usepackage{graphicx}
\graphicspath{{./}}
\usepackage{multirow}

\usepackage{booktabs}
\usepackage{xcolor}
\usepackage{enumitem}
\usepackage{hyperref}

\usepackage{subcaption}

\usepackage{amsmath}
\usepackage{amssymb}
\usepackage{mathtools}
\usepackage{amsthm}
\usepackage{tabularx}
\usepackage{pgfplots}
\usepgfplotslibrary{groupplots}
\pgfplotsset{compat=1.18}
\definecolor{revblue}{RGB}{0,82,155}
\theoremstyle{plain}
\newtheorem{theorem}{Theorem}[section]

\newtheorem{lemma}[theorem]{Lemma}
\newtheorem{corollary}[theorem]{Corollary}
\theoremstyle{definition}

\newtheorem{assumption}[theorem]{Assumption}
\theoremstyle{remark}
\newtheorem{remark}[theorem]{Remark}

\usepackage{natbib}

\usepackage{colortbl}
\usepackage{wrapfig}

\usepackage{mathpazo}

\definecolor{bestbg}{RGB}{221,235,228}

\newcommand{\best}[1]{\cellcolor{bestbg}\textbf{#1}}
\newcommand{\second}[1]{\underline{#1}}
\newcommand{\sideann}[1]{%
  \makebox[0pt][l]{\hspace{0.25em}{\scriptsize\textcolor{black}{#1}}}%
}
\newcommand{\oursup}[2]{%
  \cellcolor{bestbg}\textbf{#1}\sideann{$\uparrow$#2}%
}
\newcommand{\oursdown}[2]{%
  #1\sideann{$\downarrow$#2}%
}

\usepackage[linesnumbered,ruled,vlined]{algorithm2e}
\usepackage{algorithmic}
\usepackage{listings}
\lstdefinelanguage{yaml}{
  keywords={true, false, null, yes, no},
  keywordstyle=\color{blue}\bfseries,
  sensitive=false,
  comment=[l]{\#},
  commentstyle=\color{gray}\ttfamily,
  stringstyle=\color{red}\ttfamily,
  morestring=[b]',
  morestring=[b]"
}
\lstdefinestyle{configstyle}{
  language=yaml,
  backgroundcolor=\color{gray!8},
  basicstyle=\ttfamily\footnotesize,
  breaklines=true,
  frame=single,
  rulecolor=\color{black!30},
  numbers=left,
  numberstyle=\tiny\color{gray},
  numbersep=5pt,
  xleftmargin=15pt,
  framexleftmargin=10pt,
  columns=flexible,
  keepspaces=true,
  showstringspaces=false,
}

\definecolor{hl1}{RGB}{255,240,200} 
\definecolor{hl2}{RGB}{220,245,255} 
\definecolor{hl3}{RGB}{230,255,230} 
\definecolor{hl4}{RGB}{255,225,235} 
\definecolor{hl5}{RGB}{235,230,255} 

\newcommand{\HLtxtA}[1]{\colorbox{hl1}{\ttfamily\strut\detokenize{#1}}}
\newcommand{\HLtxtB}[1]{\colorbox{hl2}{\ttfamily\strut\detokenize{#1}}}
\newcommand{\HLtxtC}[1]{\colorbox{hl3}{\ttfamily\strut\detokenize{#1}}}
\newcommand{\HLtxtD}[1]{\colorbox{hl4}{\ttfamily\strut\detokenize{#1}}}
\newcommand{\HLtxtE}[1]{\colorbox{hl5}{\ttfamily\strut\detokenize{#1}}}

\lstdefinestyle{configstyleHL}{
  style=configstyle,
  moredelim=**[is][\HLone]{@@}{@@},
  moredelim=**[is][\HLtwo]{||}{||},
  moredelim=**[is][\HLthree]{<<}{>>},
  moredelim=**[is][\HLfour]{==}{==},
  moredelim=**[is][\HLfive]{!!}{!!}
}

\newcommand{\Prb}{\mathbb{P}}
\newcommand{\R}{\mathbb{R}}
\newcommand{\E}{\mathbb{E}}
\DeclarePairedDelimiter{\norm}{\lVert}{\rVert}
\usepackage{pifont}
\newcommand{\cmark}{\textcolor{green}{\ding{51}}}
\newcommand{\xmark}{\textcolor{red}{\ding{55}}}

\usepackage[most]{tcolorbox}
\tcbuselibrary{breakable,skins}
\usepackage{fvextra}
\newtcolorbox{CodeBox}[1][]{%
  enhanced,
  breakable,
  colback=white,
  colframe=black!25,
  boxrule=0.4pt,
  arc=1.2mm,
  left=1.5mm,right=1.5mm,top=1mm,bottom=1mm,
  #1
}

\definecolor{symphonyRow}{RGB}{235,245,255}
\definecolor{symphonyRow2}{RGB}{220,235,255}
\title{Symphony-Coord: Adaptive Routing for Multi-Agent LLM Systems}
\author{
    Zhaoyang Guan\textsuperscript{1,*}
    \quad Huixi Cao\textsuperscript{2,*}
    \quad Ming Zhong\textsuperscript{3}
    \quad Yin Wang\textsuperscript{2}
    \quad Guanyu Liu\textsuperscript{6}
    \quad Eric Yang\textsuperscript{5}
    \\
    Lynn Ai\textsuperscript{5}
    \quad Yongxin Ni\textsuperscript{4,$\dagger$}
    \quad Bill Shi\textsuperscript{5,$\dagger$}
}

\affiliation{
    \textsuperscript{1}Engineering Sciences and Applied Mathematics, Northwestern University\\
    \textsuperscript{2}New York University
    \textsuperscript{3}Independent Researcher
    \quad \textsuperscript{4}National University of Singapore
    \quad \textsuperscript{5}Gradient\\
    \textsuperscript{6}University of Macau
    \textsuperscript{*}Equal contribution.
}

\date{Jan 30, 2026}

\correspondence{
    Yongxin Ni: niyongxin@u.nus.edu
    \quad Bill Shi: tianyu.s@gradient.network
}

\project{Bill Shi}
\sourcecode{https://github.com/GradientHQ/symphony-coord}

\begin{document}

\abstract{Multi-agent large language model systems can tackle complex multi-step tasks
by decomposing work and coordinating specialized behaviors. However, current
coordination mechanisms typically rely on statically assigned roles and
centralized controllers. As agent pools and task distributions evolve, these
design choices can lead to inefficient routing, poor adaptability, and fragile
fault recovery. We introduce Symphony-Coord, a task-local coordination framework with
decentralized execution that transforms agent selection into an online
multi-armed bandit problem. Instead of relying on a fixed task-to-role map,
Symphony-Coord allows routing specializations to emerge from interaction and
feedback. The framework employs a two-stage dynamic beacon protocol:(i) a lightweight candidate screening mechanism to limit communication and
computation overhead; and (ii) an adaptive LinUCB selector that routes subtasks
using context features derived from task requirements and agent states, updated
through delayed post-execution feedback. Under candidate-conditional linear
bandit assumptions, we prove sublinear regret bounds for the immediate-feedback
selector and explicitly separate the deferred-update effects introduced by
post-vote rewards.
Validation through simulation experiments and real-world large language model
benchmarks shows that Symphony-Coord improves task routing efficiency and
recovery behavior under distribution shifts and agent failures.}
\maketitle
\raggedbottom

\section{Introduction}
\label{sec:intro}
The rapid progress of large language models (LLMs) has fueled growing interest in multi-agent systems, where multiple LLM agents decompose and collaborate to solve complex, multi-step reasoning tasks~\citep{guo2024llmsurvey}. Early multi-agent frameworks such as AutoGen and CrewAI enable building LLM applications via multiple agents that converse and collaborate, typically instantiated through an explicit orchestration structure (e.g., manager/controller patterns) to coordinate interaction and execution~\citep{wu2023autogen,crewai2024}. Related LLM-agent software engineering systems similarly employ structured multi-agent communication and coordination pipelines~\citep{qian2024chatdev,hong2023metagpt}. However, as the agent pool grows and agent heterogeneity increases, orchestration can face engineering and systems challenges (e.g., communication overhead, coordination complexity, and robustness concerns), which are widely discussed as open challenges in LLM-based multi-agent systems~\citep{guo2024llmsurvey}.

A promising alternative is \emph{decentralized execution with task-local coordination}: worker agents remain independent executors, while a lightweight router coordinates the candidate set, dispatch, and feedback update for each task instance. This is distinct from a fully centralized controller that prescribes a global workflow, but it is also not free-form peer negotiation. Across both centralized and partially decentralized designs, many systems still rely on \emph{static, pre-defined roles}. For example, role-specialized agents are a common design choice in multi-agent software development and collaboration frameworks~\citep{qian2024chatdev,hong2023metagpt}. This discrete setting is convenient to implement but often does not match real-world capabilities: the effectiveness of the same agent varies with the task context, input distribution, and system operating conditions~\citep{guo2024llmsurvey}. Therefore, static roles typically lead to three types of problems. \textbf{(1)} Routing efficiency deteriorates because these role labels are often too general and fail to distinguish subtle differences in the specific capabilities of the models. \textbf{(2)} When task distribution changes, the system struggles to switch calls to a more suitable agent in a timely manner because the mapping from roles to tasks is fixed during the design phase. \textbf{(3)} When a highly capable agent experiences performance degradation or becomes temporarily unavailable, static allocation lacks a fast and reliable replacement mechanism, easily leading to continuous performance decline~\citep{guo2024llmsurvey}.

\begin{figure}[!ht]
\centering
\includegraphics[width=1\linewidth]{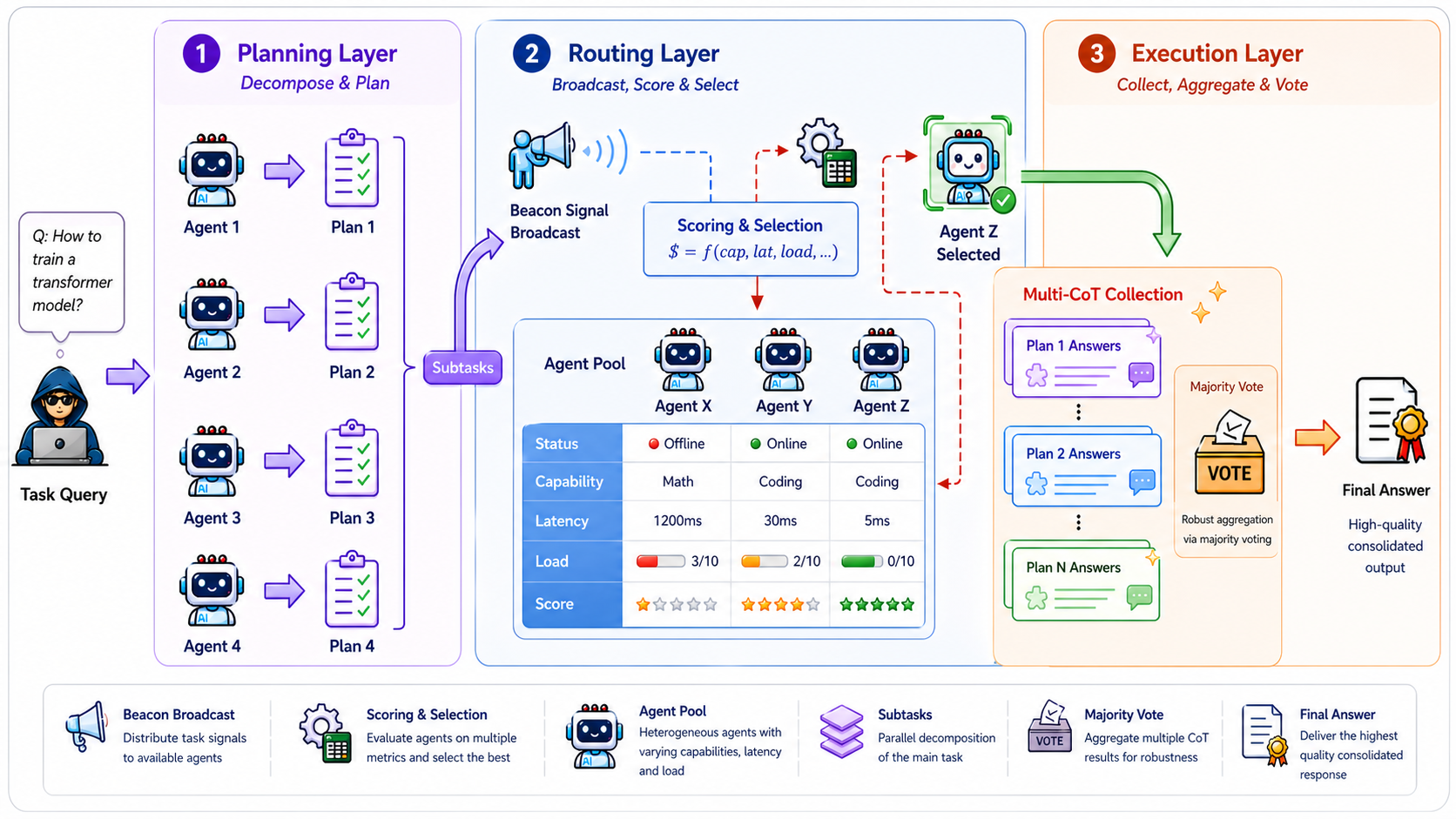}
\caption{Symphony-Coord Overview}
\label{fig:overview}
\end{figure}
Some dialogue-based frameworks (e.g., CAMEL) attempt to soften the rigidity of static roles via prompt-based ``role playing'' mechanisms (e.g., inception prompting) for autonomous cooperation among communicative agents~\citep{li2023camel}. However, externally imposed personas can be unstable and do not necessarily translate into sustained execution quality. More importantly, such approaches often lack a principled decision mechanism: the system does not learn from experience \emph{which agent is most suitable for a subtask under the current context and system state}, which remains a central challenge for scalable LLM multi-agent design~\citep{guo2024llmsurvey}.

In this paper, \emph{emergent coordination} refers to runtime-induced routing specialization rather than to manually assigned planner/solver/checker identities or unconstrained agent negotiation. Candidate agents may differ ex ante through backbones, capability tags, or prompt-conditioned variants, but Symphony-Coord does not bind these candidates to a fixed task-to-role map. The division of labor is instead expressed by the router's allocation pattern as it adapts to task features, runtime state, and feedback.

Our key perspective is that \emph{agent selection should be treated as an online decision-making problem}. When a task arrives, the system selects an executor based on observable context, then updates its routing policy after receiving feedback, aligning with contextual bandit formulations~\citep{li2010contextual,lattimore2020bandit}. This naturally involves an \emph{exploration--exploitation} trade-off that is fundamental in bandit learning: exploitation favors agents with higher estimated reward, while exploration tests uncertain candidates to reduce uncertainty and detect non-stationarity~\citep{bubeck2012regret,lattimore2020bandit}. This paradigm reduces the risk of repeatedly assigning tasks to degraded agents due to inertia, while also avoiding long-term budget waste on persistently suboptimal candidates~\citep{bubeck2012regret,lattimore2020bandit}.

Based on this idea, we propose Symphony-Coord, a decentralized multi-agent framework that models agent selection as an \emph{online contextual bandit} problem~\citep{li2010contextual,lattimore2020bandit}. Symphony-Coord introduces a two-stage \emph{dynamic Beacon} routing protocol to achieve both scalability and adaptivity. In Stage~1, the system issues lightweight Beacon queries and performs \textbf{Top-$L$ candidate filtering} using prior capability signals and constraints, sharply reducing the candidate set to control communication and compute costs (a concern emphasized in multi-agent surveys)~\citep{guo2024llmsurvey}. In Stage~2, Symphony-Coord applies a \textbf{LinUCB}-style online selector over the shortlisted candidates~\citep{li2010contextual}. The selection score combines: (a) the current estimated reward (\emph{exploitation}) and (b) an uncertainty-driven upper-confidence bonus (\emph{exploration}), and is updated online after execution feedback~\citep{li2010contextual,lattimore2020bandit}. Concretely, Symphony-Coord constructs a context vector for each candidate by fusing task requirement features, the agent's static capability representation, and dynamic runtime signals (i.e., latency, load, reliability, availability, cost, throughput, and quota / rate limit). After task completion, corresponding feedback of these signals is used to update the LinUCB statistics~\citep{li2010contextual}, enabling adaptive routing that remains robust under task distribution shifts and agent failures/degradation (motivated by multi-agent challenges and robustness considerations)~\citep{guo2024llmsurvey}. Accordingly, Symphony-Coord should be viewed as a task-local online routing formulation and system architecture for heterogeneous LLM-agent pools, rather than a standalone new bandit algorithm. Our contributions are as follows:
\begin{enumerate}[leftmargin=*, itemsep=2pt, topsep=2pt]
\item \textbf{Problem formulation.} We formalize subtask–agent routing as an online contextual bandit problem, treating each agent as an arm and constructing context features that integrate task requirements and agent states.
\item \textbf{Two-stage dynamic Beacon routing.} We propose a practical routing pipeline: Top-$L$ candidate filtering reduces communication and inference overhead (a key systems challenge in LLM multi-agent settings), while LinUCB performs adaptive selection within the retained candidate set with an explicit exploration bonus.
\item \textbf{Online learning closed loop.} We design a feedback-to-update mechanism that continuously improves routing decisions end-to-end, aligning with the sequential decision-making view of contextual bandits.
\item \textbf{Theory scope.} Under standard candidate-conditional linear bandit assumptions, we prove a high-probability sublinear regret bound for the selector and make explicit how Top-$L$ filtering and delayed post-vote feedback change the interpretation of this bound through filtering and deferred-update terms.

\item \textbf{Empirical validation.} We evaluate Symphony-Coord in simulation, real LLM benchmarks, matched-budget routing controls, component ablations, and routing diagnostics, following the common evaluation emphasis in LLM-based multi-agent research.
\end{enumerate}

\section{Related Work}
\label{sec:related}

Our work builds upon two primary streams of research: multi-agent systems for LLMs and online learning for resource allocation.

\subsection{Multi-Agent Systems for LLMs}

The landscape of LLM-based multi-agent systems has evolved rapidly, with recent surveys providing comprehensive overviews \citep{guo2024llmsurvey}. Centralized systems like AutoGen \citep{wu2023autogen} and MetaGPT \citep{hong2023metagpt} utilize a central controller to manage workflows, which simplifies coordination but introduces scalability bottlenecks and single points of failure. CrewAI \citep{crewai2024} similarly adopts centralized task management with predefined roles. ChatDev \citep{qian2024chatdev} introduced a chat-powered framework for software engineering with specialized agents communicating through structured pipelines. Nonetheless, it relies on static role assignments, such as designer and tester.

Recent work has explored more flexible agent collaboration. The CAMEL framework \citep{li2023camel} introduced inception prompting to enable autonomous cooperation between communicative agents, highlighting the potential for agents to complete complex tasks with minimal human intervention. GPTSwarm \citep{zhuge2024gptswarm} proposed graph-based agent composition with dynamic routing, while AFLOW \citep{ma2024aflow} explored automated workflow generation for multi-agent systems. Similarly, MorphAgent \citep{lu2024morphagent} introduced a decentralized framework that empowers agents to continuously refine their roles through self-evolving profiles and runtime feedback. Symphony-Coord is closest in spirit to this line of adaptive coordination work, but differs in its focus on subtask-level executor selection as an online routing problem with an explicit Top-$L$ prescreen, runtime-state features, and post-execution feedback. Thus, the novelty is not a new base bandit algorithm; it is the integration of capability-aware filtering, contextual selection, and feedback-driven routing specialization for heterogeneous LLM-agent pools. For a qualitative, framework-level comparison of representative systems and routing designs, see Appendix~\ref{Comparsion}.

\subsection{Contextual Bandits for Resource Allocation}

The multi-armed bandit framework provides a principled approach to sequential decision-making under uncertainty \citep{lattimore2020bandit, bubeck2012regret}. Contextual bandits extend this paradigm by conditioning arm selection on observable features, with LinUCB \citep{li2010contextual} being widely adopted due to its computational efficiency and strong theoretical properties. Upper Confidence Bound (UCB) algorithms \citep{auer2002ucb} provide frequentist guarantees with explicit exploration bonuses, forming the theoretical foundation for our approach.

Recent work has applied bandit algorithms to resource allocation in distributed systems \citep{ding2013resource}, model selection in contextual stochastic bandits \citep{pacchiano2020model}, and adaptive computing \citep{lattimore2020bandit}. However, the application to LLM-based multi-agent orchestration presents unique challenges: heterogeneous agents with varying capabilities, high-dimensional context spaces combining task semantics with runtime states, delayed post-execution feedback, and multi-objective rewards balancing correctness, latency, and cost. We address this setting with a two-stage protocol that integrates capability-aware filtering with LinUCB-based online selection, and we state the resulting regret analysis as a candidate-conditional guarantee for the selector rather than as a complete model of all nonlinear LLM-agent interactions.

\section{Methodology}
\label{sec:method}

The core innovation of Symphony-Coord is the transformation of agent selection from a static matching problem into an adaptive online
learning process. We achieve this through a two-stage protocol that balances efficiency and adaptivity, allowing routing specialization
to form dynamically from interactions between agents and tasks.

\subsection{Problem Formulation: Agent Selection as a Two-Stage Contextual Bandit}
\label{sec:formulation}

We formulate executor agent selection in Symphony-Coord as a two-stage contextual bandit process aligned with our
implementation. Upon receiving a user task, the system broadcasts it to $k$ planning agents to obtain a decomposition plan
that yields a sequence of subtasks. The bandit decision happens at the subtask execution level: at each time step $t$, a task arrives and the system selects one executor agent from the currently available pool
$\mathcal{A}_t \subseteq \{1,\dots,N\}$ to execute it.

\paragraph{Feedback signal and objective.}
The online selector is updated with a shaped reward signal. In our implementation, the reward
combines: (i) a base success indicator for valid, non-error completion, (ii) a winner bonus when an execution agrees with the
voted final output for the current task, (iii) an optional correctness term when gold labels are available, and (iv) an efficiency penalty based on latency.
This feedback can be delayed: under multi-path execution and voting in Appendix~\ref{sec:exec-vote}, we record step-level context features during execution and construct rewards only after the task-level vote is resolved.
The winner bonus is therefore an agreement-based proxy rather than a correctness oracle: it is capped, combined with validity and latency signals, and further examined in the reward-sensitivity study in Section~\ref{sec:post_vote_reward}.
This distinction is important because majority agreement may reinforce common wrong answers when agents share the same bias.
Symphony-Coord mitigates this risk by freezing the selector during held-out evaluation, using gold correctness only in adaptation regimes where labels are available, and supporting deployment-time verifier gates, diversity monitoring, and down-weighting of the agreement bonus when consensus conflicts with external checks.
Our objective is to learn a selection policy that maximizes cumulative expected reward:
\begin{equation}
\max \sum_{t=1}^{T} \mathbb{E}[r_t].
\label{eq:objective}
\end{equation}

During held-out test evaluation, the selector is frozen: no bandit update is performed on test samples, and gold labels are used only offline to compute reported accuracy. Thus, Eq.~\eqref{eq:objective} and the shaped reward in Appendix~\ref{sec:post_vote_reward} describe the training and online-adaptation mechanism, not a test-time use of ground-truth labels.

\paragraph{Stage 1: Top-$L$ candidate filtering.}
Querying and comparing all agents at every step is expensive. We therefore first compute a lightweight \emph{composite} score for each
available agent $j \in \mathcal{A}_t$:

\begin{equation}
\begin{aligned}
s_{j,t}
&= w_1\,\mathrm{match\_score}(j,t)
 + w_2\,\mathrm{prior\_success}(j) \\
&\quad + w_3\,\mathrm{reliability}(j,t).
\end{aligned}
\label{eq:stage1_score}
\end{equation}

where $\mathrm{match\_score}(j,t)$ measures task--agent matching via two paths.
By default, we form an agent capability text from the agent's declared capabilities extracted
from its system prompt when available, or from capability tags returned by a capability manager, and form a subtask text by
concatenating the subtask requirement with its input. We then compute cosine similarity between their embeddings and rescale it to $[0,1]$.
When embeddings are disabled or unavailable, $\mathrm{match\_score}(j,t)$ falls back to a lexical similarity computed between the subtask requirement and the agent capability text.

$\mathrm{prior\_success}(j)$ is a smoothed prior estimate of historical success, while
$\mathrm{reliability}(j,t)$ captures short-horizon operational stability that is not implied by semantic matching or historical success alone. Concretely, it is computed from recent runtime signals such as successful completion, timeout/error status, output-contract validity, availability, and latency stability; this makes an agent with strong average accuracy but recent service degradation less likely to dominate the prescreen. In practice, the routing
module implements these factors as a configurable composite scoring function; depending on configuration, some reliability penalties may be folded into the matching/prior terms rather than appearing as a strictly separate component. The weights $(w_1,w_2,w_3)$ are
configurable and fixed for the reported runs. We then construct a small candidate set by selecting the Top-$L$ agents:

\begin{equation}
C_t = \mathrm{Top}\text{-}L\big(\{s_{j,t}\}_{j\in\mathcal{A}_t}\big).
\label{eq:topL}
\end{equation}

This step prunes the decision space so that subsequent online learning focuses only on the most promising candidates. The theoretical guarantee in Appendix~\ref{Algorithm} is therefore candidate-conditional: it controls learning regret within $C_t$ when the best feasible executor is retained. When Stage~1 excludes that executor, the excess loss is captured by the filtering-regret decomposition in Remark~\ref{rem:filter}; empirically, Table~\ref{tab:topl_poolsize} reports oracle recall and exclusion gap for this prescreen.

\paragraph{Stage 2: Online contextual bandit selection within candidates.}
Given the filtered candidate set $C_t$, we construct a context vector $x_{j,t}$ for each candidate $j\in C_t$, and apply an online learner to select the
executor:
\begin{equation}
a_t \in \arg\max_{j\in C_t} \Big( x_{j,t}^\top \hat{\theta}_{t-1}
+ \beta_{t-1}\sqrt{x_{j,t}^\top A_{t-1}^{-1} x_{j,t}} \Big).
\label{eq:stage2_ucb}
\end{equation}
In Symphony-Coord, we record step-level context features for the chosen action during execution, and update the learner once the
corresponding reward becomes available; in planner mode a base per-step update may be applied before voting, otherwise updates
are post-hoc after voting in Appendix~\ref{sec:exec-vote}.
Delayed feedback changes the time at which the statistics are updated, but not the information available when the action is selected: routing uses only the current subtask context and runtime state, and the post-vote signal is applied to logged step records after the corresponding task outcome is resolved. Equivalently, the selector uses an observed-reward history that excludes unresolved decisions; Appendix~\ref{Algorithm} states the immediate-feedback theorem and then identifies the extra staleness term induced by these pending rewards.

\subsection{LinUCB-based Dynamic Selection}
\label{sec:linucb}

Within the candidate set $C_t$, we apply LinUCB to make the final choice. For each candidate agent $j\in C_t$, we construct a context
vector:
\begin{equation}
x_{j,t} = \big[1,\ \mathrm{sim\_emb}(j,t),\ \ell_{j,t},\ \tau_{j,t},\ \rho_{j,t},\ u_{j,t}\big]^{\top},
\label{eq:context_vec}
\end{equation}
where $\mathrm{sim\_emb}(j,t)$ is the same similarity signal used in Stage~1 matching. It is the cosine similarity between embeddings of the agent capability text and the subtask text
formed by concatenating \texttt{requirement} with its prompt. $\ell_{j,t}$ is the current load, $\tau_{j,t}$ is the normalized latency, $\rho_{j,t}$ is the historical
reputation, and $u_{j,t}$ is an availability flag.
The selector uses the standard linear reward model:

\begin{equation}
\mathbb{E}[r_t \mid a_t=j] = x_{j,t}^{\top}\theta^{*},
\label{eq:linear_reward}
\end{equation}
and maintains a ridge-regression estimator with
\begin{equation}
\begin{aligned}
A_t &= \lambda I + \sum_{s=1}^{t-1} x_{a_s,s}x_{a_s,s}^{\top}, \\
b_t &= \sum_{s=1}^{t-1} r_s x_{a_s,s}, \\
\hat{\theta}_t &= A_t^{-1} b_t .
\end{aligned}
\label{eq:ridge_state}
\end{equation}
where $\lambda>0$ is the regularization parameter. For each $j\in C_t$, the UCB score is computed as
\begin{equation}
\mathrm{UCB}_{j,t}
= x_{j,t}^{\top}\hat{\theta}_t
+ \beta_t \sqrt{x_{j,t}^{\top}A_t^{-1}x_{j,t}},
\label{eq:ucb}
\end{equation}
where $\beta_t$ controls exploration. The selected agent is
\begin{equation}
a_t = \arg\max_{j\in C_t} \mathrm{UCB}_{j,t}.
\label{eq:select}
\end{equation}
In the online adaptation phase, after the feedback reward becomes available, LinUCB updates

\begin{equation}
\begin{aligned}
A_{t+1} &\leftarrow A_t + x_{a_t,t} x_{a_t,t}^{\top}, \\
b_{t+1} &\leftarrow b_t + r_t x_{a_t,t}, \\
\hat{\theta}_{t+1} &\leftarrow A_{t+1}^{-1} b_{t+1}.
\end{aligned}
\label{eq:update}
\end{equation}

Together, the Top-$L$ filtering and LinUCB selection form an efficient-yet-adaptive routing mechanism: Stage~1 reduces overhead, while Stage~2 improves the policy from feedback, balancing exploitation and exploration under non-stationarity. The full selection-and-update procedure is summarized in Algorithm~\ref{alg:linucb_beacon} in Appendix~\ref{app:algorithm}. The theory in Appendix~\ref{app:algorithm} should be read as a selector-level analysis under linear realizability and candidate retention assumptions, complemented by empirical checks of filtering recall, reward sensitivity, and domain-shift behavior.

\paragraph{System instantiation.}
To make the system boundary explicit, Symphony-Coord uses decentralized worker execution with a task-local routing plane. The router creates Beacon queries, forms $C_t$, dispatches the selected executor, records step-level features, and applies post-vote updates for the current task instance; workers execute assigned subtasks independently and do not run open-ended negotiation to choose the next executor. In our implementation, the two-stage router is embedded into an end-to-end workflow with optional task decomposition, multi-path execution, final-key normalization, voting, and post-vote reward shaping for delayed credit assignment. Full component roles, message flow, and the exact weighting definitions shared by match\_score and sim\_emb are provided in Appendix~\ref{app:system}.

\section{Experiment}
\subsection{Main Results}
Table~\ref{tab:main-comp-noacc} compares Symphony-Coord with single-agent and multi-agent baselines on GSM8K, BBH, and MedicalQA. We report accuracy for all tasks and provide evaluation details in Appendix~\ref{app:exp}. Green cells mark the best method within each backbone--benchmark column.Overall, Symphony-Coord consistently improves over single-agent baselines across all backbones, with gains of 8.5--22.0 on GSM8K, 16.5--23.5 on BBH, and 27.0--33.0 on MedicalQA. Against multi-agent baselines, it achieves the best average accuracy under every backbone, outperforming the strongest competitor by 1.0--4.7 points on average. While some baselines win on individual benchmarks, Symphony-Coord is more balanced across domains. The weak cold-start results further show the need for adaptive routing rather than fixed or naive multi-agent allocation.

\begin{table}[!ht]
  \caption{
  Performance comparison of different methods across various tasks.}
  \label{tab:main-comp-noacc}
  \centering
  \footnotesize
  \setlength{\tabcolsep}{4pt}
  \renewcommand{\arraystretch}{1.05}
  \setlength{\aboverulesep}{0pt}
  \setlength{\belowrulesep}{0pt}
  \setlength{\cmidrulesep}{0.3pt}

  \begin{tabular}{l l l | c | c | c}
    \toprule
    \textbf{Backbone} & \textbf{Category} & \textbf{Method} &
    \textbf{GSM8K (ACC\%)} & \textbf{BBH (ACC\%)} & \textbf{MedicalQA (ACC\%)} \\
    \midrule

    \multirow{11}{*}{DeepSeek-V3} &
      \multirow{5}{*}{Single Agent} &
      Direct            & 58.00 & 52.00 & 45.00 \\
    & & React             & 58.00 & 56.00 & 49.00 \\
    & & Synapse           & 60.00 & 56.50 & 53.00 \\
    & & Self-Consistency  & 67.00 & 67.50 & 48.00 \\
    & & Self-Refinement   & 68.50 & 69.50 & 50.00 \\
    \cmidrule(lr){2-6}
    & \multirow{6}{*}{Multi-Agent} &
      Morphagent         & 47.00 & 58.00 & 47.00 \\
    & & Cold Start        & 43.00 & 55.00 & 50.00 \\
    & & MetaGPT           & \best{81.00} & 65.00 & \textbf{84.00} \\
    & & AFLOW             & 73.00 & 84.00 & 80.00 \\
    & & GPTSwarm          & 75.00 & \best{88.00} & 83.00 \\
    & & \textbf{Symphony-Coord} & \oursdown{\textbf{77.00}}{4.00} & \oursdown{\textbf{86.00}}{2.00} & \oursup{86.00}{2.00} \\
    \midrule

    \multirow{11}{*}{DeepSeek-V3-0324} &
      \multirow{5}{*}{Single Agent} &
      Direct            & 62.00 & 55.50 & 54.00 \\
    & & React             & 62.00 & 51.00 & 53.00 \\
    & & Synapse           & 60.00 & 57.50 & 51.00 \\
    & & Self-Consistency  & 61.00 & 55.00 & 53.00 \\
    & & Self-Refinement   & 62.00 & 53.50 & 53.00 \\
    \cmidrule(lr){2-6}
    & \multirow{6}{*}{Multi-Agent} &
      Morphagent         & 41.00 & 54.00 & 64.00 \\
    & & Cold Start        & 46.00 & 53.00 & 48.00 \\
    & & MetaGPT           & \textbf{78.00} & 72.00 & \textbf{80.00} \\
    & & AFLOW             & 69.00 & \best{83.00} & \textbf{80.00} \\
    & & GPTSwarm          & 73.00 & 79.00 & 76.00 \\
    & & \textbf{Symphony-Coord} & \oursup{84.00}{6.00} & \oursdown{\textbf{81.00}}{2.00} & \oursup{81.00}{1.00} \\
    \midrule

    \multirow{11}{*}{GPT-5-nano} &
      \multirow{5}{*}{Single Agent} &
      Direct            & 36.00 & 56.50 & 51.00 \\
    & & React             & 37.00 & 61.00 & 44.00 \\
    & & Synapse           & 51.00 & 57.50 & 47.00 \\
    & & Self-Consistency  & 62.00 & 65.00 & 43.00 \\
    & & Self-Refinement   & 68.00 & 59.50 & 47.00 \\
    \cmidrule(lr){2-6}
    & \multirow{6}{*}{Multi-Agent} &
      Morphagent         & 47.00 & 46.50 & 42.00 \\
    & & Cold Start        & 43.00 & 36.50 & 17.00 \\
    & & MetaGPT           & \textbf{77.00} & 69.50 & \best{82.00} \\
    & & AFLOW             & 69.00 & 76.50 & \textbf{77.00} \\
    & & GPTSwarm          & 76.00 & \textbf{81.00} & 73.00 \\
    & & \textbf{Symphony-Coord} & \oursup{78.00}{1.00} & \oursup{83.50}{2.50} & \oursup{82.00}{5.00} \\
    \bottomrule
  \end{tabular}
\end{table}

\subsection{Experiments on Heterogeneous Agents}
\label{sec:hetero_medicalqa}
\begin{wrapfigure}{r}{0.50\textwidth}
\vspace{-8pt}
\centering
\footnotesize
\begin{tikzpicture}
\begin{axis}[
    ybar,
    bar width=6pt,
    width=\linewidth,
    height=3.65cm,
    ymin=75, ymax=88,
    enlarge x limits=0.20,
    ylabel={Acc. (\%)},
    ylabel style={font=\scriptsize},
    symbolic x coords={FH,SH,LH,BH},
    xtick=data,
    xticklabels={
        {Fully\\Homo.},
        {Skill\\Hetero.},
        {LLM\\Hetero.},
        {Both\\Hetero.}
    },
    x tick label style={
        font=\scriptsize\rmfamily,
        rotate=0,
        align=center
    },
    y tick label style={font=\scriptsize},
    ymajorgrids=true,
    grid style={dashed, gray!45},
    axis x line*=bottom,
    axis y line*=left,
    tick align=outside,
    legend style={
        at={(0.5,1.04)},
        anchor=south,
        legend columns=2,
        draw=none,
        font=\scriptsize,
        /tikz/every even column/.append style={column sep=3pt}
    },
    nodes near coords,
    every node near coord/.append style={font=\tiny\rmfamily},
]

\addplot[fill=revblue!35, draw=revblue] coordinates {
    (FH,79.00)
    (SH,80.00)
    (LH,82.00)
    (BH,83.00)
};

\addplot[fill=revblue!75, draw=revblue] coordinates {
    (FH,80.00)
    (SH,86.00)
    (LH,81.00)
    (BH,84.00)
};

\legend{3 Agents, 5 Agents}
\end{axis}
\end{tikzpicture}
\vspace{-6pt}
\caption{MedicalQA accuracy under different heterogeneity settings.}
\label{fig:medicalqa_hetero_bar}
\vspace{-10pt}
\end{wrapfigure}
We evaluate Symphony-Coord on MedicalQA with three- and five-agent teams under four pool settings: fully homogeneous, skill-heterogeneous, LLM-heterogeneous, and both-heterogeneous.
Figure~\ref{fig:medicalqa_hetero_bar} shows that heterogeneity helps more as the team grows. With three agents, gains are small, while LLM and combined heterogeneity are slightly stronger. With five agents, skill heterogeneity reaches the highest accuracy, suggesting that capability specialization is the key source of complementarity.

\subsection{Ablation Study}
\paragraph{Component-wise Ablation of Symphony-Coord.}
Table~\ref{tab:ablation} reports MedicalQA ablations under the 200 / 300 / 100 cold-start, pre-train, and test protocol.
We fix CoT = 3 for all variants and use Plan-$K$ = 3 only when subtask decomposition is enabled. Variants without UCB
have no learnable routing update, so pre-train is marked as not applicable.
\begin{table}[!htbp]
  \centering
  \footnotesize
  \setlength{\tabcolsep}{2.2pt}
  \renewcommand{\arraystretch}{1.05}

  \caption{MedicalQA ablation study of Symphony-Coord under our three-phase protocol.}
  \label{tab:ablation}

  \begin{tabular*}{\textwidth}{@{\extracolsep{\fill}} l c c c c c c c c c c}
    \toprule
    \multirow{2}{*}{Variant} &
    \multicolumn{7}{c}{Components} &
    \multicolumn{3}{c}{MedicalQA (\%)} \\
    \cmidrule(lr){2-8}\cmidrule(lr){9-11}
    & Multi-Agent & Static & Random & Vote & Top-$L$ & UCB & Subtask
    & Cold start & Pre-train & Test \\
    \midrule

    A0 (Single agent)                     & \xmark & \xmark & \xmark & \xmark & \xmark & \xmark & \xmark & 50.00 & -- & 51.00 \\
    A1 (Random pick)                      & \cmark & \xmark & \cmark & \xmark & \xmark & \xmark & \xmark & 50.00 & -- & 52.00 \\
    A2 (Static default)                   & \cmark & \cmark & \xmark & \xmark & \xmark & \xmark & \xmark & 52.00 & -- & 46.00 \\
    A3 (Naive vote)                       & \cmark & \xmark & \xmark & \cmark & \xmark & \xmark & \xmark & 48.00 & -- & 55.00 \\
    A4 (Round-robin)                      & \cmark & \xmark & \xmark & \xmark & \xmark & \xmark & \xmark & 53.00 & -- & 48.00 \\
    \midrule

    A5 (UCB)                              & \cmark & \xmark & \xmark & \xmark & \xmark & \cmark & \xmark & 50.00 & 69.00 & 73.00 \\
    A6 (Subtask; static top-1)            & \cmark & \cmark & \xmark & \xmark & \xmark & \xmark & \cmark & 50.00 & -- & 69.00 \\
    A7 (Top-$L$; static top-1)            & \cmark & \cmark & \xmark & \xmark & \cmark & \xmark & \xmark & 57.00 & -- & 74.00 \\
    \midrule

    A8 (Top-$L$ + UCB)                    & \cmark & \xmark & \xmark & \xmark & \cmark & \cmark & \xmark & 48.50 & 77.00 & 79.00 \\
    A9 (Top-$L$ + Subtask; static)        & \cmark & \cmark & \xmark & \xmark & \cmark & \xmark & \cmark & 50.50 & -- & 77.00 \\
    A10 (UCB + Subtask)                   & \cmark & \xmark & \xmark & \xmark & \xmark & \cmark & \cmark & 53.00 & 74.67 & 79.00 \\
    \rowcolor{symphonyRow}
    \textbf{A11 (Symphony-Coord)}         & \cmark & \xmark & \xmark & \xmark & \cmark & \cmark & \cmark & 55.00 & 83.00 & 86.00 \\
    \bottomrule
  \end{tabular*}
\end{table}
Table~\ref{tab:ablation} shows that naive multi-agent routing is unreliable: random pick and voting give only minor gains,
whereas static default and round-robin can reduce accuracy. UCB is the main driver, raising test accuracy from 51.00 to
73.00. Top-$L$ filtering and subtask decomposition become most effective when paired with adaptive routing, with the full
system achieving the best cold-start, pre-train, and test results (55.00 / 83.00 / 86.00).

This improvement is not due to extra test-time compute alone. Under the same Plan-$K{=}3$ and CoT 3 budget, static
Top-$L$ + Subtask reaches 77.00 (A9), and Top-$L$ + UCB without subtasks reaches 79.00 (A8), both below the full
86.00 result. This confirms the complementarity of filtering, learning-based routing, and subtask-aware execution.
\subsection{Further Analysis}

\paragraph{Scalability Stress Test.}
\label{sec:scalability-stress}
We evaluate scalability by increasing the candidate pool from $N=5$ to $N=100$ while fixing Symphony-Coord's shortlist size to Top-$L=3$, keeping its online routing budget constant.
\begin{figure*}[!htbp]
    \centering
    \includegraphics[width=0.98\textwidth]{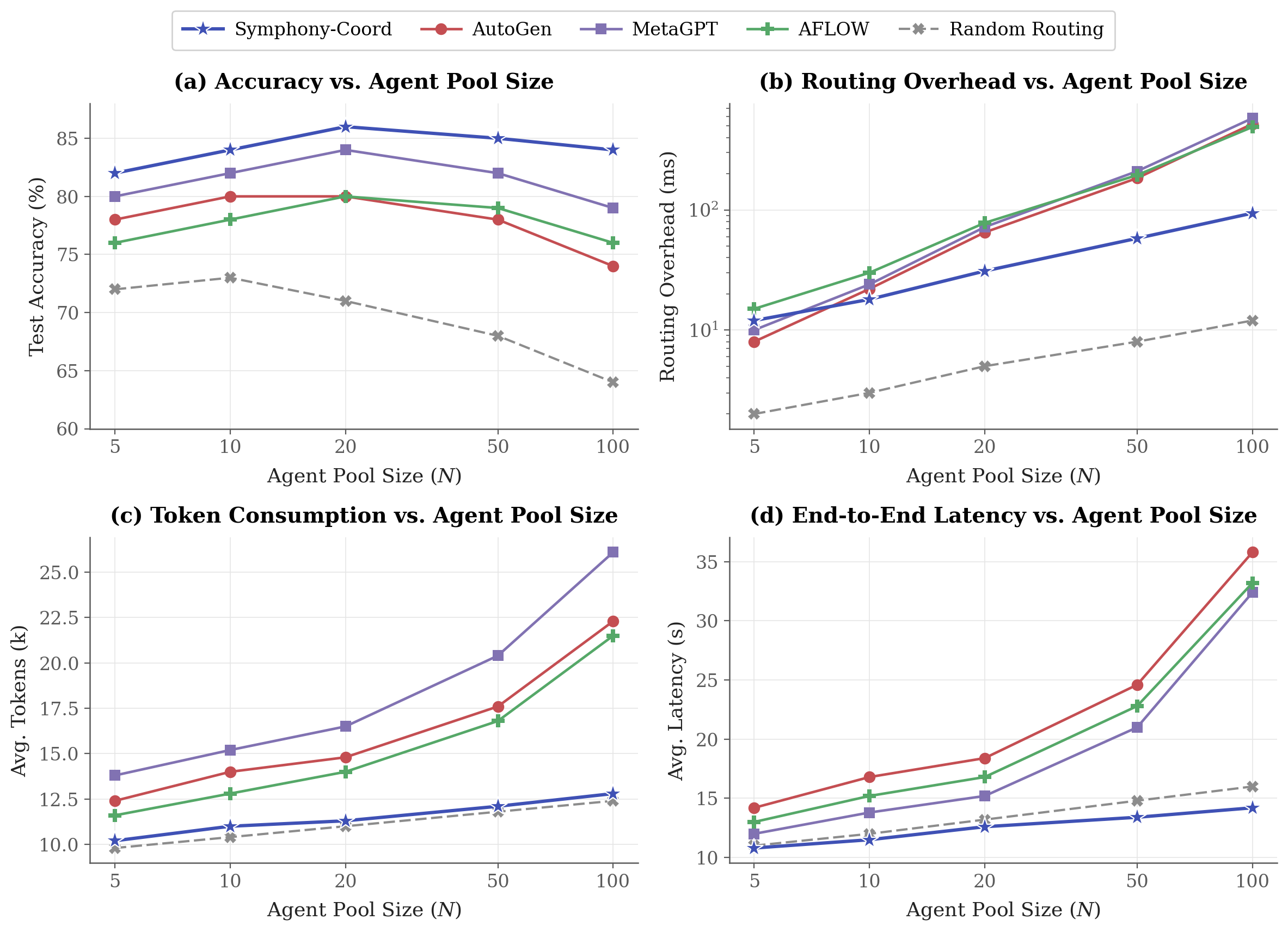}
    \caption{Scalability stress test under increasing agent pool size. Symphony-Coord fixes Top-$L=3$ and maintains a constant online routing budget as $N$ grows.}
    \label{fig:scalability_stress}
\end{figure*}
As shown in Figure~\ref{fig:scalability_stress}, Symphony-Coord remains stable as the pool grows and reaches its best accuracy at $N=20$.
In contrast, AutoGen, MetaGPT, and AFLOW suffer larger drops under large pools, suggesting that weakly filtered coordination introduces additional routing noise.
Efficiency metrics show the same pattern: baseline routing overhead increases rapidly with $N$, reaching hundreds of milliseconds at $N=100$, whereas Symphony-Coord stays below $100$ ms by selecting only from three shortlisted candidates.
Its token usage and end-to-end latency also grow more slowly, demonstrating a more favorable accuracy--cost trade-off under large heterogeneous agent pools.The detailed candidate agent composition and configuration used in this scalability stress test are provided in Appendix~\ref{app:exp}.

\paragraph{Cost--Accuracy Pareto Frontier.}
We further study cost--accuracy trade-offs on MedQA using three cost axes: average token consumption, end-to-end latency, and estimated dollar cost.
The comparison covers five LLM backbones: DeepSeek, GPT-OSS-120B, Qwen-2.5-7B, Gemini-2.5-flash-lite, and Grok-4.1-fast.
Figure~\ref{fig:cost_accuracy_pareto} shows that Symphony-Coord lies in a favorable Pareto region.
Compared with single-agent baselines, it achieves higher accuracy with a moderate inference budget; compared with multi-agent baselines, it reaches comparable or higher accuracy with lower tokens, latency, or estimated cost in most settings.
Together with the scalability stress test, these results show that Symphony-Coord remains efficient under both larger candidate pools and constrained inference budgets.
\begin{figure*}[!htbp]
    \centering
    \makebox[\textwidth][c]{%
        \includegraphics[width=1\textwidth]{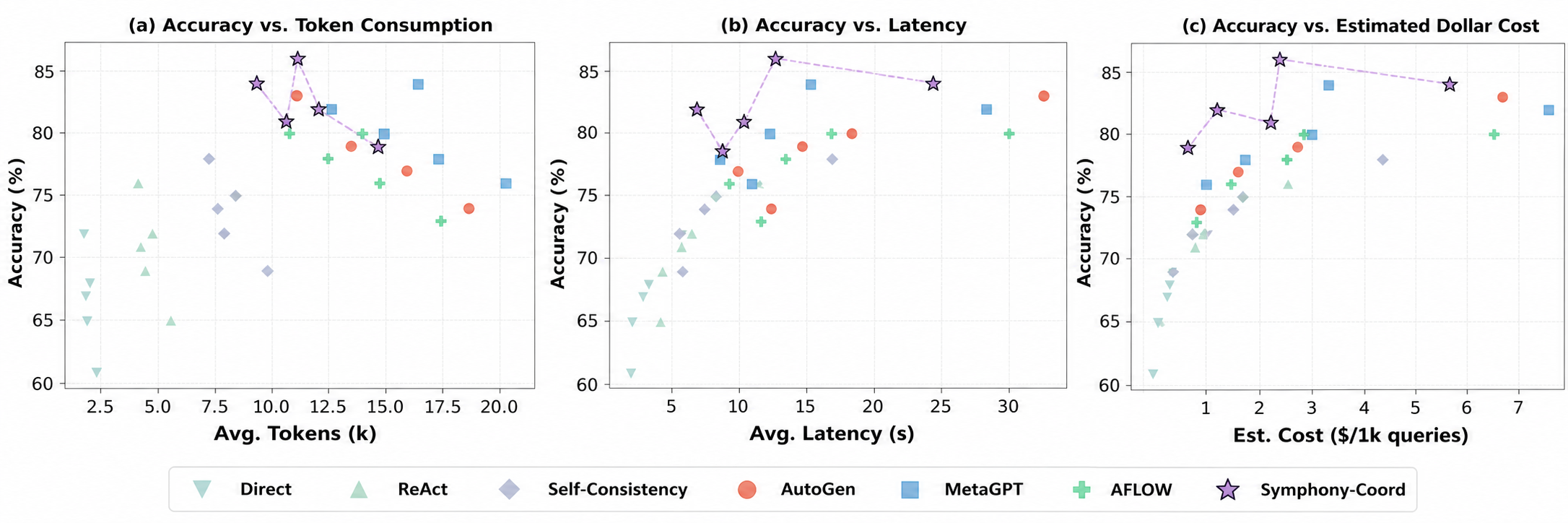}
    }
    \caption{Cost--accuracy Pareto frontier on MedQA across five LLM backbones.}
    \label{fig:cost_accuracy_pareto}
\end{figure*}

\section{Conclusion}
We presented Symphony-Coord, a task-local coordination framework for heterogeneous LLM-based multi-agent systems. Unlike conventional methods that rely on fixed roles or centralized controllers, Symphony-Coord formulates executor agent selection as an online contextual bandit problem. By combining capability-aware Top-$L$ candidate filtering, LinUCB-based contextual selection, and post-execution feedback updates, Symphony-Coord allows task-to-agent allocation to adapt dynamically to runtime states and task distributions. As a result, the division of labor is not imposed through predefined role labels, but emerges gradually as routing specialization through interaction and feedback.

On the theoretical side, we provide a candidate-conditional sublinear regret guarantee under standard linear realizability and candidate-retention assumptions, and clarify how Top-$L$ filtering and delayed post-vote feedback affect the interpretation of this guarantee. Empirically, Symphony-Coord consistently improves over strong single-agent baselines across multiple LLM backbones and tasks, while achieving more balanced overall performance compared with representative multi-agent baselines. Further heterogeneity studies, ablations, scalability stress tests, and cost--accuracy analyses show that capability filtering, online selection, and subtask-aware execution provide complementary benefits. Even as the candidate pool grows or the inference budget becomes constrained, Symphony-Coord maintains a favorable accuracy--cost trade-off.

Overall, Symphony-Coord demonstrates that scalable multi-agent LLM coordination does not need to rely on static role assignment. Instead, it can be realized through adaptive routing driven by runtime context and feedback. Future work will extend Symphony-Coord to larger and more open agent pools, richer tool-use workflows, and longer-horizon task streams, while systematically evaluating deployment-oriented factors such as latency, cost, service instability, and distribution shift.

\appendix
\onecolumn
\section{Algorithm}
\label{Algorithm}
\subsection{Setup and Notation}

We model the dynamic Beacon-based agent selection as an online contextual bandit.
At each round $t=1,2,\dots,T$, a subtask arrives with requirements summarized in a feature vector
$r_t\in\R^{d_r}$. Each agent $j\in\{1,\dots,N\}$ has a capability embedding $c_j\in\R^{d_c}$
and a dynamic state vector $z_{j,t}\in\R^{d_z}$.
The orchestrator constructs a context vector
\begin{equation}
x_{j,t} \;\triangleq\; g(r_t,c_j,z_{j,t}) \in \R^d,
\end{equation}
where $g(\cdot)$ is a fixed feature map.

\subsubsection{Two-stage candidate filtering}

To keep communication and computation efficient, Symphony-Coord first builds a candidate set using a
capability prior score $\phi(c_j,r_t)\in[0,1]$. Let
\begin{equation}
\mathcal{C}_t \;\triangleq\; \text{Top-}L\big(\phi(c_j,r_t)\big),
\end{equation}
and optionally enforce feasibility constraints like:deadline, availability:
\begin{equation}
\mathcal{C}_t \leftarrow \{j\in\mathcal{C}_t:\ d_{j,t}\le D_t,\ u_{j,t}=1\}.
\end{equation}
The dynamic selector then chooses an executor $a_t\in\mathcal{C}_t$ using a learning rule.

\subsubsection{Linear reward model}

We assume a standard linear contextual bandit model on the candidate set.

\begin{assumption}[Bounded contexts]
For all $t$ and all $j\in\mathcal{C}_t$, $\norm{x_{j,t}}_2 \le 1$.
\end{assumption}

\begin{assumption}[Linear realizability and sub-Gaussian noise]
There exists an unknown parameter $\theta^\star\in\R^d$ with $\norm{\theta^\star}_2 \le S$ such that the
\emph{expected} reward of choosing agent $j$ at time $t$ is
\begin{equation}
\mu_{j,t} \;\triangleq\; \E[R_t \mid a_t=j, \mathcal{F}_{t-1}] \;=\; x_{j,t}^\top \theta^\star.
\end{equation}
The observed reward is
\begin{equation}
R_t = x_{a_t,t}^\top \theta^\star + \varepsilon_t,
\end{equation}
where $\varepsilon_t$ is conditionally $ \sigma^2$-sub-Gaussian given $\mathcal{F}_{t-1}$:
$\E[\exp(\lambda \varepsilon_t)\mid \mathcal{F}_{t-1}] \le \exp(\lambda^2\sigma^2/2)$ for all $\lambda\in\R$.
\end{assumption}

\begin{assumption}Candidate contains the instantaneous optimal arm. Let $j_t^\star \in \arg\max_{j\in\{1,\dots,N\}} \mu_{j,t}$ be an optimal agent at time $t$.
Assume $j_t^\star \in \mathcal{C}_t$ for all $t$.
\end{assumption}

\subsubsection{LinUCB-based Dynamic Beacon Selection}

Define (ridge) covariance and response accumulators:
\begin{equation}
A_0 = \lambda I_d,\qquad b_0 = 0\in\R^d,
\end{equation}
and for $t\ge 1$,
\begin{equation}
A_t = A_{t-1} + x_{a_t,t}x_{a_t,t}^\top,\qquad b_t=b_{t-1}+R_t x_{a_t,t}.
\end{equation}
The ridge estimator is $\hat{\theta}_t = A_t^{-1} b_t$.

At round $t$, for each $j\in\mathcal{C}_t$, compute
\begin{equation}
\mathrm{UCB}_{j,t} \;=\; x_{j,t}^\top \hat{\theta}_{t-1} \;+\; \beta_{t-1}\sqrt{x_{j,t}^\top A_{t-1}^{-1}x_{j,t}},
\end{equation}
and select
\begin{equation}
a_t \in \arg\max_{j\in\mathcal{C}_t} \mathrm{UCB}_{j,t}.
\end{equation}

\subsection{Regret Bound for LinUCB Dynamic Beacon Selection}

We first state a standard self-normalized confidence bound used in linear bandit analysis.

\begin{lemma}[High-probability confidence ellipsoid]\label{lem:conf}
Fix $\delta\in(0,1)$. Under Assumptions 1--2, with probability at least $1-\delta$,
for all $t\ge 0$,
\begin{equation}\label{eq:ellipsoid}
\norm{\hat{\theta}_t - \theta^\star}_{A_t} \;\le\;
\beta_t \;\triangleq\;
\sigma \sqrt{2\log\!\frac{1}{\delta} + d\log\!\Big(1+\frac{t}{\lambda}\Big)}
\;+\; \sqrt{\lambda}\,S.
\end{equation}
\end{lemma}

\begin{remark}
Lemma~\ref{lem:conf} is a standard result (self-normalized martingale concentration) in linear bandits.
A complete proof can be found in classical references such as \cite{abbasi2011improved}.
Here we use it as a known lemma to keep the appendix focused on the selection mechanism.
\end{remark}

\begin{lemma}[One-step regret bound]\label{lem:onestep}
On the event of Lemma~\ref{lem:conf}, at any round $t\ge 1$,
\begin{equation}\label{eq:onestep}
\mu_{j_t^\star,t} - \mu_{a_t,t}
\;\le\;
2\beta_{t-1}\sqrt{x_{a_t,t}^\top A_{t-1}^{-1} x_{a_t,t}}.
\end{equation}
\end{lemma}

\begin{proof}
Fix $t$ and let $\mathcal{E}$ be the event in Lemma~\ref{lem:conf}.
On $\mathcal{E}$, for any $j\in\mathcal{C}_t$,
\begin{align}
\mu_{j,t}
&= x_{j,t}^\top \theta^\star
= x_{j,t}^\top \hat{\theta}_{t-1} + x_{j,t}^\top(\theta^\star-\hat{\theta}_{t-1}) \nonumber\\
&\le x_{j,t}^\top \hat{\theta}_{t-1} + \norm{x_{j,t}}_{A_{t-1}^{-1}}\norm{\theta^\star-\hat{\theta}_{t-1}}_{A_{t-1}}
\le x_{j,t}^\top \hat{\theta}_{t-1} + \beta_{t-1}\norm{x_{j,t}}_{A_{t-1}^{-1}} \nonumber\\
&= \mathrm{UCB}_{j,t}. \label{eq:optimistic}
\end{align}
In particular, using Assumption 3, $j_t^\star\in\mathcal{C}_t$ so $\mu_{j_t^\star,t}\le \mathrm{UCB}_{j_t^\star,t}$.
By definition of $a_t$,
$\mathrm{UCB}_{a_t,t}\ge \mathrm{UCB}_{j_t^\star,t}$, hence
\begin{equation}
\mu_{j_t^\star,t} - \mu_{a_t,t}
\le
\mathrm{UCB}_{a_t,t} - \mu_{a_t,t}.
\end{equation}
Finally, on $\mathcal{E}$ we also have a lower confidence bound:
\[
\mu_{a_t,t}
= x_{a_t,t}^\top \theta^\star
\ge x_{a_t,t}^\top \hat{\theta}_{t-1} - \beta_{t-1}\norm{x_{a_t,t}}_{A_{t-1}^{-1}},
\]
so
\[
\mathrm{UCB}_{a_t,t} - \mu_{a_t,t}
\le
\big(x_{a_t,t}^\top \hat{\theta}_{t-1} + \beta_{t-1}\norm{x_{a_t,t}}_{A_{t-1}^{-1}}\big)
-
\big(x_{a_t,t}^\top \hat{\theta}_{t-1} - \beta_{t-1}\norm{x_{a_t,t}}_{A_{t-1}^{-1}}\big)
=2\beta_{t-1}\norm{x_{a_t,t}}_{A_{t-1}^{-1}}.
\]
This proves \eqref{eq:onestep}.
\end{proof}

\begin{lemma}[Elliptical potential]\label{lem:elliptic}
Let $A_0=\lambda I_d$ and $A_t=A_{t-1}+x_tx_t^\top$ with $\norm{x_t}_2\le 1$.
Then
\begin{equation}\label{eq:elliptic}
\sum_{t=1}^T \min\Big\{1,\, x_t^\top A_{t-1}^{-1}x_t \Big\}
\;\le\;
2\log\frac{\det(A_T)}{\det(A_0)}
\;\le\;
2d\log\Big(1+\frac{T}{\lambda}\Big).
\end{equation}
\end{lemma}

\begin{theorem}[High-probability regret bound for dynamic Beacon (LinUCB)]\label{thm:regret}
Under Assumptions 1--3, fix $\delta\in(0,1)$. With probability at least $1-\delta$,
the cumulative regret
\[
\mathrm{Reg}(T)\triangleq \sum_{t=1}^T \big(\mu_{j_t^\star,t} - \mu_{a_t,t}\big)
\]
satisfies
\begin{equation}\label{eq:regretbound}
\mathrm{Reg}(T)
\;\le\;
2\beta_{T}\sqrt{2T\,d\,\log\Big(1+\frac{T}{\lambda}\Big)}.
\end{equation}
Consequently, $\E[\mathrm{Reg}(T)] = \tilde{O}(d\sqrt{T})$ (up to logarithmic factors).
\end{theorem}

\begin{proof}
Work on the event $\mathcal{E}$ of Lemma~\ref{lem:conf} (probability $\ge 1-\delta$), and note $\beta_{t-1}\le \beta_T$.
From Lemma~\ref{lem:onestep},
\[
\mu_{j_t^\star,t} - \mu_{a_t,t}
\le 2\beta_T\sqrt{x_{a_t,t}^\top A_{t-1}^{-1} x_{a_t,t}}.
\]
Summing over $t$ and applying Cauchy--Schwarz:
\begin{align}
\mathrm{Reg}(T)
&\le 2\beta_T \sum_{t=1}^T \sqrt{x_{a_t,t}^\top A_{t-1}^{-1} x_{a_t,t}} \nonumber\\
&\le 2\beta_T \sqrt{T \sum_{t=1}^T x_{a_t,t}^\top A_{t-1}^{-1} x_{a_t,t}}. \label{eq:cs}
\end{align}
Since $x_{a_t,t}^\top A_{t-1}^{-1} x_{a_t,t}\le \norm{x_{a_t,t}}_2^2/\lambda \le 1/\lambda$, we can bound
$\sum x^\top A^{-1}x$ by $\sum \min\{1, x^\top A^{-1}x\}$ up to constants. Using Lemma~\ref{lem:elliptic} with $x_t=x_{a_t,t}$,
\[
\sum_{t=1}^T x_{a_t,t}^\top A_{t-1}^{-1} x_{a_t,t}
\le
\sum_{t=1}^T \min\Big\{1,\, x_{a_t,t}^\top A_{t-1}^{-1} x_{a_t,t}\Big\}\cdot \max\Big\{1, \frac{1}{\lambda}\Big\}.
\]
For simplicity (and standard in the literature), take $\lambda\ge 1$ so the max factor is $1$; otherwise it contributes only a constant.
Then Lemma~\ref{lem:elliptic} yields
\[
\sum_{t=1}^T x_{a_t,t}^\top A_{t-1}^{-1} x_{a_t,t}
\le 2d\log\Big(1+\frac{T}{\lambda}\Big).
\]
Plugging into \eqref{eq:cs} proves \eqref{eq:regretbound}.
\end{proof}

\begin{remark}[Delayed feedback]
Theorem~A.8 analyzes the standard \emph{immediate-feedback} LinUCB setting.
In the full Symphony-Coord system, however, rewards may become available only
after post-vote aggregation across multiple execution paths. In this case,
the context of each routing decision is logged at decision time, while the
parameter update is applied only when the corresponding reward is resolved.

Delayed feedback does not change the information available at action-selection
time and does not introduce later information into the routing decision. Instead,
it postpones estimator updates, so its main effect is stale parameters and
consequently slower adaptation. For this reason, Theorem~A.8 should be interpreted
as a clean baseline guarantee for the selection module under immediate feedback,
while the full system corresponds to a deferred-update variant of the same
contextual bandit selector.

To make this gap explicit, let $d_t$ denote the number of unresolved rewards at
round $t$, and let $A^{\mathrm{obs}}_{t-1}$ be the covariance matrix formed from
resolved rewards only. The implemented selector replaces $A_{t-1}$ in
Eq.~\eqref{eq:ridge_state} with $A^{\mathrm{obs}}_{t-1}$, and therefore acts as
LinUCB on a delayed history. Its regret can be decomposed into the
immediate-feedback selector regret plus an additional staleness term caused by
decisions made before earlier rewards are incorporated:
\[
\mathrm{Reg}_{\mathrm{delayed}}(T)
\le
\mathrm{Reg}_{\mathrm{immediate}}(T)
+
\sum_{t=1}^{T}
\bigl[\mu_{\tilde a_t,t}-\mu_{a_t,t}\bigr]_+ ,
\]
where $\tilde a_t$ is the action that the same selector would take if all rewards
up to $t-1$ had already resolved, and $[\cdot]_+$ denotes the positive part.
This term is small when post-vote latency is bounded or when the number of
pending updates is kept small, and it grows with stale-decision exposure rather
than with access to additional label information. In implementation, this is
controlled by logging every routing decision, applying rewards in resolution
order, and optionally batching or throttling new updates when the pending queue
becomes large.

A formal regret analysis that explicitly incorporates delayed post-vote feedback
is beyond the scope of the current paper and is left for future work.
\end{remark}

\begin{remark}\label{rem:filter}
If Assumption 3 does not always hold, define the \emph{exclusion loss}
\[
\Delta^{\mathrm{filter}}_t \triangleq \mu_{j_t^\star,t} - \max_{j\in\mathcal{C}_t} \mu_{j,t} \;\ge 0.
\]
Then the total regret decomposes as
\[
\mathrm{Reg}(T) \le \underbrace{\sum_{t=1}^T \Delta^{\mathrm{filter}}_t}_{\text{filtering regret}}
+
\underbrace{\sum_{t=1}^T \Big(\max_{j\in\mathcal{C}_t}\mu_{j,t}-\mu_{a_t,t}\Big)}_{\text{learning regret within candidates}}.
\]

Theorem~\ref{thm:regret} controls the second term; the first term is small when the Top-$L$ filter
includes high-quality agents with high probability.
\end{remark}

\subsection{One-shot Top-1 Mis-selection Probability (Sub-Gaussian Scores)}

This section proves a classical exponential bound on selecting the wrong top agent when each score is
a noisy observation of the true utility.

\begin{assumption}[One-shot noisy score model]\label{ass:oneshot}
There are $K$ candidate agents with deterministic utilities $u_1,\dots,u_K$.
We observe $s_j = u_j + \xi_j$ where $\xi_j$ are independent, mean-zero, $\sigma^2$-sub-Gaussian.
Let $j^\star = \arg\max_j u_j$ be the unique best agent and $\Delta_j = u_{j^\star}-u_j>0$ for $j\ne j^\star$.
\end{assumption}

\begin{theorem}[Exponential mis-selection bound]\label{thm:misselect}
Under Assumption~\ref{ass:oneshot},
\begin{equation}\label{eq:misselect}
\Prb\!\left(\arg\max_{j} s_j \neq j^\star\right)
\;\le\;
\sum_{j\ne j^\star}\exp\!\left(-\frac{\Delta_j^2}{4\sigma^2}\right).
\end{equation}
\end{theorem}

\begin{proof}
By a union bound,
\[
\Prb\!\left(\arg\max_j s_j \neq j^\star\right)
\le
\sum_{j\ne j^\star} \Prb(s_j \ge s_{j^\star}).
\]
Fix $j\ne j^\star$. Note
\[
s_j \ge s_{j^\star}
\iff
(u_j+\xi_j) - (u_{j^\star}+\xi_{j^\star}) \ge 0
\iff
\xi_j - \xi_{j^\star} \ge \Delta_j.
\]
Since $\xi_j$ and $\xi_{j^\star}$ are independent $\sigma^2$-sub-Gaussian, the difference
$\xi_j-\xi_{j^\star}$ is mean-zero and $(2\sigma^2)$-sub-Gaussian. Hence, for any $z>0$,
\[
\Prb(\xi_j-\xi_{j^\star} \ge z)\le \exp\!\left(-\frac{z^2}{4\sigma^2}\right).
\]
Setting $z=\Delta_j$ gives $\Prb(s_j \ge s_{j^\star}) \le \exp(-\Delta_j^2/(4\sigma^2))$ and summing over $j\ne j^\star$
yields \eqref{eq:misselect}.
\end{proof}

\subsection{Overall Task Quality Guarantee from Regret}

We now translate the regret bound into a bound on total expected task quality (cumulative reward).

Define the \emph{oracle cumulative expected reward}:
\[
\mathrm{OPT}(T)\triangleq \sum_{t=1}^T \mu_{j_t^\star,t},
\]
and the algorithm's cumulative expected reward:
\[
\mathrm{ALG}(T)\triangleq \sum_{t=1}^T \mu_{a_t,t}.
\]
By definition, $\mathrm{Reg}(T)=\mathrm{OPT}(T)-\mathrm{ALG}(T)$.

\begin{corollary}[Near-oracle cumulative task quality]\label{cor:quality}
Under the assumptions of Theorem~\ref{thm:regret}, with probability at least $1-\delta$,
\[
\mathrm{ALG}(T) \ge \mathrm{OPT}(T) - 2\beta_{T}\sqrt{2T\,d\,\log\Big(1+\frac{T}{\lambda}\Big)}.
\]
In particular, dividing by $T$ yields an average per-round suboptimality of order
$\tilde{O}\big(d/\sqrt{T}\big)$.
\end{corollary}

\begin{proof}
Immediate from $\mathrm{ALG}(T)=\mathrm{OPT}(T)-\mathrm{Reg}(T)$ and Theorem~\ref{thm:regret}.
\end{proof}

\begin{remark}
If the reward encodes task success probability, Corollary~\ref{cor:quality} states that
the long-run average task quality approaches the oracle that always selects the best agent each round.
\end{remark}

\subsection{Robustness to Capability Drift}

We provide two commonly used non-stationarity models and show how the dynamic selector can be adapted.

\subsection{Piecewise-stationary model (change points)}

Assume $\theta^\star$ is \emph{piecewise constant} over time: there exist change points
$1=\tau_0 < \tau_1 < \cdots < \tau_M < \tau_{M+1}=T+1$ such that for all $t\in[\tau_m,\tau_{m+1})$,
\[
\mu_{j,t} = x_{j,t}^\top \theta^{(m)}
\]
for some fixed $\theta^{(m)}$ in segment $m$.
Suppose (for analysis) that the algorithm can \emph{reset} LinUCB at each change point.

\begin{theorem}[Regret under $M$ change points with oracle resets]\label{thm:changepoints}
Under Assumptions 1--2 within each stationary segment, and assuming the optimal agent is always in the candidate set,
the regret of ``Reset-LinUCB'' satisfies, with probability at least $1-\delta$,
\[
\mathrm{Reg}(T) \;\le\;
2\beta_T \sqrt{2d\log\!\Big(1+\frac{T}{\lambda}\Big)}\cdot \sqrt{(M+1)T}.
\]
Equivalently, $\mathrm{Reg}(T)=\tilde{O}\big(d\sqrt{(M+1)T}\big)$.
\end{theorem}

\begin{proof}
Let segment $m$ have length $T_m=\tau_{m+1}-\tau_m$, with $\sum_{m=0}^M T_m = T$.
Applying Theorem~\ref{thm:regret} to each segment (with a union bound over segments for probability),
we have $\mathrm{Reg}_m \le C \beta_T \sqrt{T_m d\log(1+T/\lambda)}$ for a universal constant $C$.
Summing and applying Cauchy--Schwarz:
\[
\mathrm{Reg}(T) = \sum_{m=0}^M \mathrm{Reg}_m
\le C\beta_T \sqrt{d\log(1+T/\lambda)}\sum_{m=0}^M \sqrt{T_m}
\le C\beta_T \sqrt{d\log(1+T/\lambda)}\sqrt{(M+1)\sum_{m=0}^M T_m},
\]
which equals the stated bound (absorbing constants into the leading factor).
\end{proof}

\begin{remark}
When change points are unknown, practical systems use change detection, sliding windows, or discounting.
The theorem above provides a clean baseline showing graceful degradation with the number of regime shifts $M$.
\end{remark}
\begin{figure}[!ht]
    \centering
    \includegraphics[width=1\linewidth]{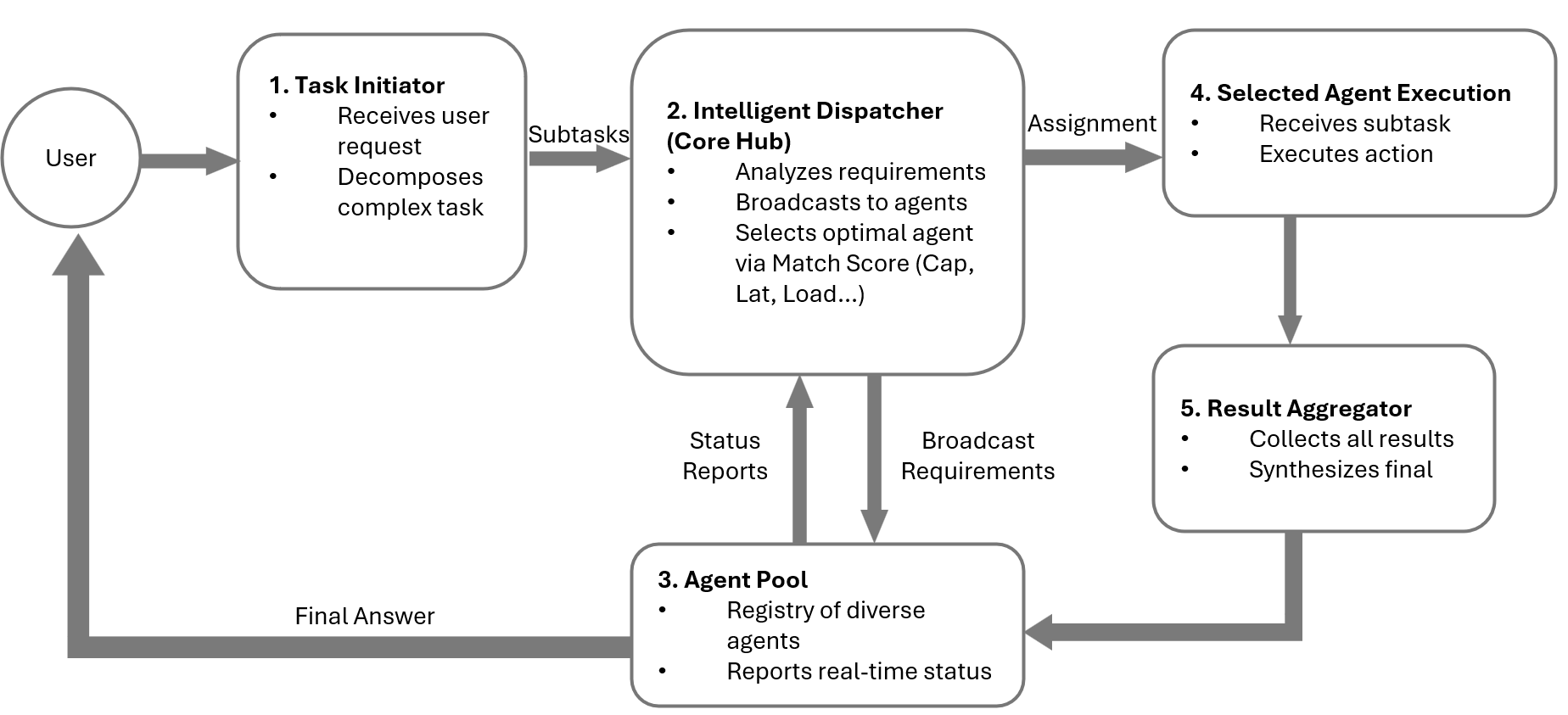}
    \caption{System Architecture}
    \label{fig:systemarch}
\end{figure}
\subsubsection{Gradual drift model}

Now assume a time-varying parameter $\theta_t^\star$ and define the total variation budget
\[
V_T \triangleq \sum_{t=2}^T \norm{\theta_t^\star - \theta_{t-1}^\star}_2.
\]
A standard approach is \emph{Sliding-Window LinUCB}: estimate parameters using only the most recent $W$ samples.
Intuitively, smaller $W$ tracks drift better but increases estimation noise; larger $W$ does the opposite.

\begin{theorem}[Dynamic regret decomposition for sliding-window LinUCB (informal but standard)]\label{thm:drift}
Assume Assumptions 1--2 hold with $\theta^\star$ replaced by $\theta_t^\star$ (bounded by $S$) and
candidate sets contain the instantaneous optimal arm.
Then sliding-window LinUCB with window size $W$ achieves a dynamic regret bound of the form
\[
\mathrm{Reg}_{\mathrm{dyn}}(T)
\;\le\;
\tilde{O}\big(d\sqrt{T W}\big)
\;+\;
O\!\Big(\frac{T}{W}\, V_T\Big),
\]
where the first term is an estimation/exploration term and the second term is a tracking (drift) term.
Choosing $W \asymp (T/V_T)^{2/3}$ yields the typical rate
$\mathrm{Reg}_{\mathrm{dyn}}(T)=\tilde{O}\big(d\,T^{2/3}V_T^{1/3}\big)$.
\end{theorem}

\section{System Architecture and Execution Pipeline}
\label{app:system}

This appendix describes how the core mechanisms in the main text are instantiated in our system implementation. Figure~\ref{fig:systemarch} provides an end-to-end overview of the execution pipeline. The system boundary is task-local: a routing layer constructs Beacons, maintains the candidate statistics, dispatches selected executors, and aggregates feedback for the current task instance, while executor agents independently perform the assigned subtasks.

\subsection{Execution and Voting: Multi-Path CoT with Post-Vote Updates}
\label{sec:exec-vote}
We expose two orthogonal controls for robustness and modularity:
(i) whether to decompose an incoming task into a sequence of subtasks using Plan-$K$, and
(ii) how many independent executions to run for each subtask using CoT.
When Plan-$K{=}1$, the input is treated as a single subtask and executed directly.
When Plan-$K{>}1$, the system generates multiple decomposition plans; each plan yields a subtask chain
$\mathrm{chain}_k=\{t_{k,1},\dots,t_{k,L_k}\}$.

For each subtask $t$, we run $P=\mathrm{CoT}$ independent executions, and record step-level context features during execution for post-hoc learning updates.
Across subtasks, the system passes a running memory of previous results to preserve dependencies.
LLM outputs may vary in formatting, so we first extract a canonical final answer.
If a structured final key is present, we extract the marked segment; otherwise we
fallback to the normalized raw output. For run $i$, this yields
\begin{equation}
y_i = \mathrm{Norm}(\mathrm{Extract}(\hat{y}_i)).
\end{equation}
By default, we apply majority voting across the $P$ runs for the current task:
\begin{equation}
S(y)=\sum_{i=1}^{P}\mathbb{I}[y_i=y],\qquad
y^\star=\arg\max_{y} S(y).
\label{eq:vote_majority}
\end{equation}
When multiple plans are generated, we optionally enable a weighted variant of voting across plans.
Each plan $k$ is assigned a plan weight computed as the average Stage-1 matching score along its subtask chain:
\begin{equation}
w_k=\frac{1}{|\mathrm{chain}_k|}\sum_{m\in \mathrm{chain}_k} s_{k,m},
\label{eq:plan_weight}
\end{equation}
where $s_{k,m}$ denotes the match score of the selected executor at step $m$ when executing plan $k$. The final plan-level output is selected by weighted voting:
\begin{equation}
S(y)=\sum_{k=1}^{K} w_k\,\mathbb{I}[y_k=y],\qquad
y^\star=\arg\max_{y} S(y).
\label{eq:vote_weighted}
\end{equation}
When CoT is used together with Plan-$K$, we first aggregate multiple runs within each plan to obtain a plan-level
output $y_k$, and then apply weighted voting across plans using $w_k$.
Note that $w_k$ is a heuristic aggregation weight derived from Stage-1 matching.

\subsection{Post-Vote Rewards and Delayed Credit Assignment}
\label{sec:post_vote_reward}

After selecting the voted final output $y^\star$ for the current task, we construct a run-level shaped reward $r_i$ and apply it
to the step records from run $i$ to perform online updates. The shaped reward combines:
(i) a base validity term, (ii) a winner bonus if the run agrees with the voted final, (iii) an optional correctness term when gold
labels are available, and (iv) a latency-based efficiency penalty:
\begin{equation}
r_i
=
\mathbb{I}[\mathrm{valid}(\hat{y}_i)]
+
\mathbb{I}[y_i=y^\star]\cdot b_{\mathrm{win}}
+
\mathbb{I}[\mathrm{gold}]\cdot
\Big(\mathbb{I}[\mathrm{correct}(y^\star)]\cdot b_{\mathrm{corr}}
-\mathbb{I}[\neg\mathrm{correct}(y^\star)]\cdot p_{\mathrm{inc}}\Big)
-
\lambda_{\mathrm{lat}}\sqrt{\mathrm{lat\_norm}(\hat{y}_i)}.
\label{eq:shaped_reward_app}
\end{equation}

Importantly, when gold labels are available during the adaptation phase, correctness is evaluated on the voted final output $y^\star$, and this correctness signal is broadcast to all steps via the post-vote update. In the held-out benchmark test phase reported in the main tables, the selector is frozen and this update rule is not applied to test samples; gold labels are used only for offline accuracy reporting.

The agreement term $\mathbb{I}[y_i=y^\star]\cdot b_{\mathrm{win}}$ is intentionally treated as a bounded proxy rather than a standalone reward. Majority agreement can be misleading when several agents share a correlated error pattern. Therefore, the reported configuration uses agreement together with other reward components: the base validity term filters malformed or contract-violating outputs, the latency penalty discourages slow consensus-seeking behavior, and the optional correctness term can override pure agreement when labels are available during adaptation. In deployments without labels, the same reward interface can incorporate rule-based verifiers, retrieval checks, unit tests, calibrated confidence thresholds, or human/audit review before committing the update. The router can also monitor answer diversity and reduce $b_{\mathrm{win}}$ when consensus becomes too concentrated or conflicts with an external verifier. These mechanisms reduce the risk that a majority-vote signal could otherwise train the router toward common but wrong answers.

\subsection{Online Update Triggering}
\label{sec:update_trigger}

We update the online selector using the step records collected during execution and the shaped reward constructed post-vote.

In the default CoT path, updates are applied after the subtask-level vote is resolved. This deferred update implements delayed credit assignment: the selected action is recorded at routing time, and the corresponding reward is attached only after output normalization and voting complete. Unresolved actions remain in a pending queue and are not used in the ridge estimator until their rewards become available. This prevents post-vote information from changing the historical action record, while making the theoretical gap a matter of stale updates rather than information leakage.
\subsection{Implementation Detail: Matching Similarity Used by Routing and Voting Weights}
\label{sec:match_score_impl}
The Stage-1 matching signal used throughout the system follows a two-path implementation:
\begin{itemize}
  \item \textbf{Embedding path.}
  We form an agent capability text from the agent's declared capabilities, and form a subtask text by concatenating the requirement with the subtask input description. We then compute cosine similarity between their embeddings and rescale it to $[0,1]$.
  \item \textbf{Lexical fallback.}
  We fall back to a normalized lexical similarity computed between the subtask requirement and the agent capability text.
  The prompt portion is not used in this fallback path.
\end{itemize}

This definition is shared by (i) $\mathrm{match\_score}(\cdot)$ in Stage-1 filtering, (ii) the plan-weight average in
Eq.~\eqref{eq:plan_weight}, and (iii) the similarity feature $\mathrm{sim\_emb}(\cdot)$ inside the LinUCB context vector.

The reliability term in Eq.~\eqref{eq:stage1_score} is computed from runtime status logs rather than from semantic similarity. It aggregates recent successful completion, timeout or service errors, contract-valid output, availability, and latency stability into a normalized score. This separates short-term service health from $\mathrm{prior\_success}(j)$, which summarizes longer-horizon historical task success, and from $\mathrm{match\_score}(j,t)$, which measures task--capability alignment.

\section{Comparison of Multi-Agent Frameworks}
\label{Comparsion}
\begin{figure}[!htbp]
  \centering
  \includegraphics[width=0.65\linewidth]{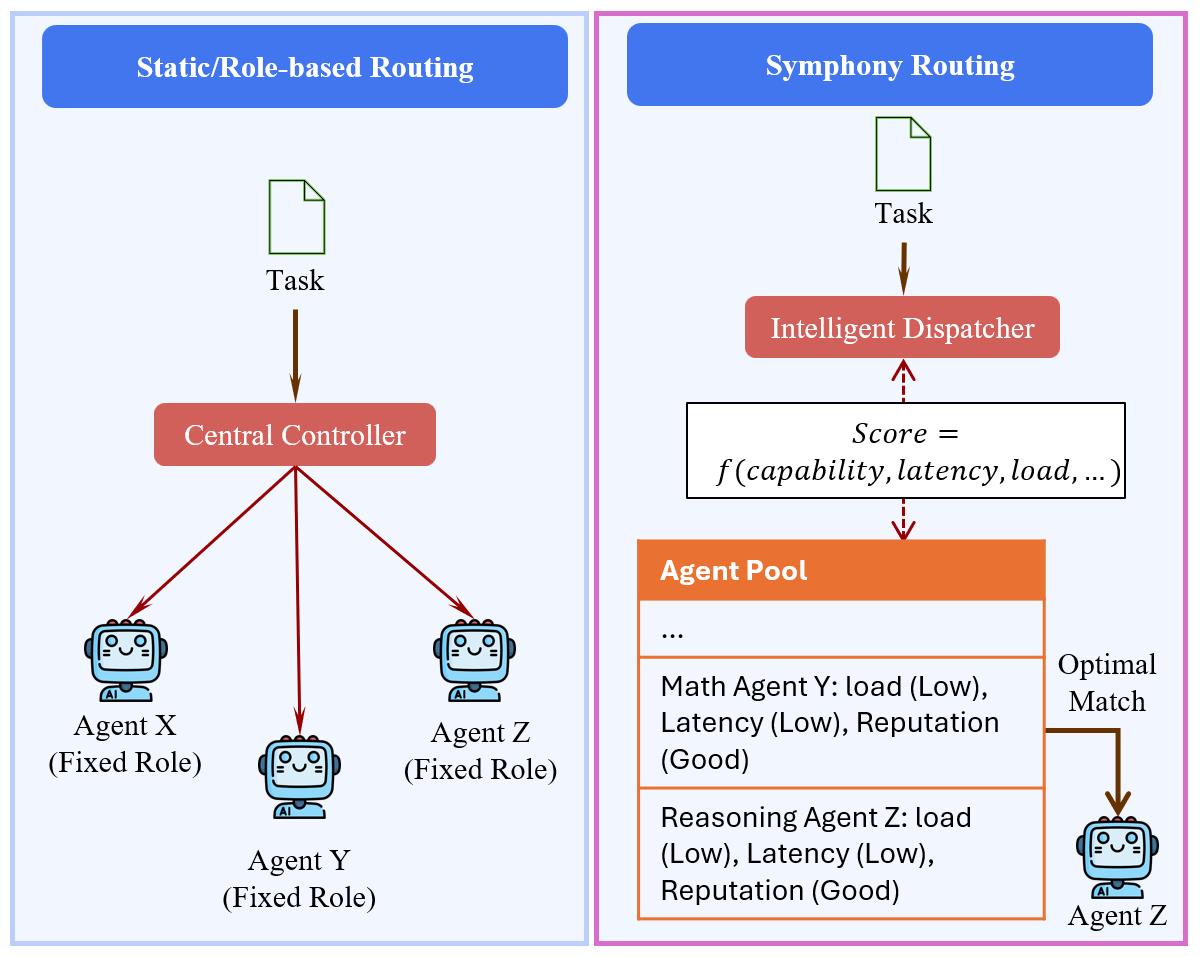}
  \caption{Comparison between static role routing and Symphony-Coord's task-local adaptive routing.}
  \label{fig:routing_compare}
\end{figure}

Figure~\ref{fig:routing_compare} contrasts Symphony-Coord with static and role-based routing approaches. Rather than binding agents to fixed task roles, Symphony-Coord uses a task-local dispatcher that dynamically adapts routing decisions to both the task context and current agent states. This design improves robustness against agent churn and reduces wasted calls by prioritizing the most promising candidates.

\paragraph{Framework-level positioning.}
\begin{table*}[!htbp]
\centering
\caption{Qualitative comparison of multi-agent orchestration frameworks.}
\label{tab:framework_comparison}
\scriptsize
\setlength{\tabcolsep}{3.5pt}
\renewcommand{\arraystretch}{1.08}

\newcolumntype{L}[1]{>{\raggedright\arraybackslash}p{#1}}

\begin{tabularx}{\textwidth}{
L{0.13\textwidth}
L{0.18\textwidth}
L{0.16\textwidth}
L{0.19\textwidth}
X
}
\toprule
\textbf{Framework}
& \textbf{Coordination style}
& \textbf{Routing unit}
& \textbf{Capability matching}
& \textbf{Adaptation / failure behavior} \\
\midrule

AutoGen
& Controller-style orchestration
& Conversation / tool workflow
& Not a primary objective
& Simple coordination, but the controller can become a bottleneck. \\

MetaGPT
& Hierarchical workflow
& Predefined roles
& Role-assignment based
& Stable role structure, but less adaptive to changing executor quality. \\

CrewAI
& Controller-style task management
& Predefined agents and tasks
& Mostly design-time
& Easy to compose, but role/task routing is largely fixed after configuration. \\

GPTSwarm
& Optimizable agent graph
& Graph nodes and edges
& Indirect via graph optimization
& Adapts graph structure, but not focused on runtime subtask executor selection. \\

AFLOW
& Automated workflow generation
& Workflow operators
& Indirect via workflow search
& Optimizes workflow structure; runtime service health is not the central signal. \\

MorphAgent
& Decentralized profile evolution
& Self-evolving agent profiles
& Profile-driven
& Refines roles/profiles through decentralized collaboration. \\

\midrule
\rowcolor{revblue!8}
\textbf{Symphony-Coord}
& \textbf{Decentralized execution with task-local routing}
& \textbf{Subtask executor selection}
& \textbf{Explicit Top-$L$ prescreen + runtime features}
& \textbf{Reallocates using feedback while keeping bounded task-local coordination.} \\

\bottomrule
\end{tabularx}
\end{table*}
Table~\ref{tab:framework_comparison} summarizes the architectural distinction between Symphony-Coord and representative multi-agent orchestration frameworks. The comparison is qualitative: it is intended to locate the source of the contribution, not to claim that one framework dominates all others across tasks. Symphony-Coord is closest to adaptive coordination systems such as GPTSwarm, AFLOW, and MorphAgent, but differs in treating subtask-level executor choice as an online routing problem with capability-aware prescreening, runtime-state features, and post-execution feedback. Thus, the contribution is not a new LinUCB algorithm or decentralization alone; it is the combination of decentralized worker execution with bounded, feedback-driven task-local routing.

\paragraph{Decentralization boundary.}
The most precise characterization of Symphony-Coord is \emph{decentralized execution with task-local coordination}. Executor agents run independently and may differ in capability, cost, latency, and availability, but the current task still has a bounded routing loop that constructs Beacons, forms the candidate set, selects executors, aggregates outputs, and applies delayed feedback updates. Table~\ref{tab:decentralization_boundary} breaks down this boundary at the component level.

\begin{table}[!htbp]
\centering
\caption{Component-level coordination boundary in Symphony-Coord.}
\label{tab:decentralization_boundary}
\footnotesize
\setlength{\tabcolsep}{3.2pt}
\renewcommand{\arraystretch}{1.03}

\newcolumntype{L}[1]{>{\raggedright\arraybackslash}p{#1}}

\begin{tabularx}{\textwidth}{
L{0.17\textwidth}
L{0.15\textwidth}
L{0.21\textwidth}
L{0.22\textwidth}
X
}
\toprule
\textbf{Component}
& \textbf{Control scope}
& \textbf{Main function}
& \textbf{Scalability effect}
& \textbf{Failure effect} \\
\midrule

Worker / executor agents
& Decentralized
& Execute assigned subtasks independently
& Enables parallel execution over heterogeneous agents
& Failures are mostly localized to assigned subtasks. \\

Capability / runtime state
& Shared task-level state
& Stores capability tags, priors, reliability, latency, load, and availability
& Cached signals reduce repeated full-pool probing
& Stale or noisy signals may weaken Stage-1 filtering. \\

Beacon prescreen
& Task-local coordination
& Builds the Top-$L$ candidate set from lightweight signals
& Bounds communication and compute before online selection
& Over-pruning may exclude a strong executor. \\

LinUCB selector / updater
& Task-local coordination
& Selects within $C_t$ and updates from delayed feedback
& Keeps online learning local to the routed task stream
& Corrupted feedback can affect current routing adaptation. \\

Voting / aggregation
& Task-local coordination
& Normalizes outputs, votes, and constructs post-vote rewards
& Avoids open-ended negotiation among workers
& Consensus errors may bias rewards; verifier gates and diversity checks mitigate this risk. \\

\bottomrule
\end{tabularx}
\end{table}
This boundary clarifies the scalability and failure claims. Symphony-Coord does not remove coordination; it makes coordination bounded and task-local. The main scalability benefit comes from not executing or comparing the full pool at every decision, while the remaining task-local router provides a controlled place to combine capability signals, runtime state, and feedback. The main failure benefit comes from separating executor failures from routing updates: worker degradation can be routed around after feedback, whereas stale candidate information or noisy consensus can still affect a task instance. This is why the Top-$L$ analysis in Appendix~\ref{sec:Supplementary Experiment} and the reward discussion in Section~\ref{sec:post_vote_reward} are part of the system interpretation rather than separate implementation details. For consensus-based reward shaping, deployment should treat agreement as one bounded signal among validity, latency, verifier, and audit signals; this prevents the router from learning a preference for agents that merely agree with common errors.

\section{Experimental Setup}
\label{app:exp}

This appendix describes the datasets, splits, routing budgets, and evaluation protocol used in the main benchmark results and supplementary system studies.

\subsection{Benchmarks}

Our main benchmark results are reported on three task families: GSM8K, BBH, and MedicalQA.
These benchmarks cover mathematical reasoning, multi-step logical reasoning, and domain-specific medical question answering, respectively.

\begin{itemize}
    \item \textbf{GSM8K} consists of grade-school math word problems that require multi-step arithmetic reasoning~\cite{cobbe2021gsm8k}.
    \item \textbf{BBH} contains challenging BIG-Bench Hard tasks that emphasize multi-step symbolic, logical, and compositional reasoning~\cite{suzgun2023bbh}.
    \item \textbf{MedicalQA} is based on MedQA-USMLE-style clinical reasoning questions and evaluates domain-specific medical knowledge and decision making~\cite{jin2021medqa}.
\end{itemize}

Additional auxiliary task sources, such as AMC and HumanEval, are used only in supplementary mixed-task or profiling studies when explicitly stated.
They are not part of the main benchmark table.

\subsection{Agent Pools and Routing Budgets}
\label{app:agent_pools}

We use pre-specified backbone and candidate-agent configurations for each experiment.
In the main benchmark results, comparisons are made under matched backbone settings: within each backbone block, all baselines and Symphony-Coord are evaluated on the same task stream, answer-normalization pipeline, and generation budget.
For system-level studies that require an explicit candidate pool, such as the MedQA cost-efficiency analysis, we use a fixed candidate set and keep it unchanged across all compared methods.
When an experiment changes the candidate-pool size, such as the scalability stress test from $N=5$ to $N=100$, this change is stated explicitly and is used only to study scaling behavior.

For Symphony-Coord, we use the same two-stage routing pipeline throughout: Stage 1 applies Top-$L$ candidate filtering, and Stage 2 performs LinUCB-based online selection within the shortlisted candidates.
Unless an experiment explicitly varies the candidate-pool size or routing budget, we fix Top-$L=3$ so that the online selector operates under a controlled and constant selection budget.

\subsection{Data Splits and Online Evaluation Protocol}

We evaluate Symphony-Coord using a three-phase protocol: cold start, pre-train, and held-out test.
The cold-start phase provides an initial stream for bootstrapping routing statistics.
During the pre-train phase, the online selector updates from execution feedback.
During the held-out test phase, the selector is frozen: no bandit update is performed on test samples, and gold labels are used only for offline accuracy computation.

For MedicalQA ablations, we use the same $200/300/100$ cold-start, pre-train, and test protocol as in the main ablation study.
For BBH, we use a fixed 200-problem test set, with 10 examples per task, and use the remaining sampled online examples for cold-start and pre-train streams.
For GSM8K and MedicalQA main evaluations, we follow the same cold-start/pre-train/test construction unless otherwise specified.

All methods are evaluated with the same answer extraction and normalization pipeline.
For multiple-choice tasks, the final answer is normalized to the option token.
For mathematical tasks, numeric answers are normalized before exact-match evaluation.
For code or structured-output auxiliary tasks, outputs are parsed according to the required schema.

\subsection{Routing and Generation Budgets}

Unless otherwise stated, Symphony-Coord uses the same two-stage routing pipeline as described in the main method.
Stage 1 applies Top-$L$ candidate filtering, and Stage 2 performs LinUCB-based online selection within the shortlisted candidates.
We use Top-$L=3$ as the default routing budget in the main and supplementary experiments, so that the online selector compares only a fixed number of shortlisted candidates rather than the full agent pool.

For decomposition and reasoning budgets, we use Plan-$K=3$ and CoT $=3$ by default.
Plan-$K$ controls the number of decomposition plans, while CoT controls the number of independent reasoning runs per subtask before normalization and voting.
Sensitivity to these two budgets is reported separately in Appendix~\ref{app:budget_sensitivity}.

\subsection{Mixed Task Pool}

For supplementary system studies, we construct a mixed task pool by combining examples from the available benchmark families into a single stream.
Each example retains its dataset identity, which is used only for analysis and diagnostic reporting.
The router does not receive the dataset label as an explicit supervision signal; it relies on task--agent matching scores, runtime state, and feedback-driven updates.

\subsection{Difficulty Tags}

We use lightweight difficulty tags only for stratified sampling and diagnostic analysis in supplementary studies.
Difficulty is not a main benchmark metric and is not used as a direct supervision signal for held-out test evaluation.
Within each dataset, we compute a simple proxy score from task-specific observable features, such as input length, reference-solution step count when available, option length, or structured-output complexity.
Scores are normalized within each dataset and then binned by percentiles:
\[
\text{bin}(d) =
\begin{cases}
\text{easy}, & d \leq P_{20}(D), \\
\text{hard}, & d \geq P_{80}(D), \\
\text{medium}, & \text{otherwise},
\end{cases}
\]
where $P_k(D)$ denotes the $k$-th percentile of the difficulty distribution within the corresponding dataset.
When a supplementary experiment requires clear separation between simple and challenging tasks, we sample from the easy and hard bins; otherwise, all difficulty bins are retained.

\section{Pseudocode}
\label{app:algorithm}
Algorithm~\ref{alg:linucb_beacon} summarizes the LinUCB-based beacon selection procedure used in Symphony-Coord.
At each round, Stage 1 first produces a Top-$L$ candidate set, and the online selector then constructs a context vector for each shortlisted agent and computes its UCB score.
The agent with the highest score is selected for execution.
Because feedback may be delayed until post-vote aggregation is completed, the algorithm updates the LinUCB parameters only when resolved rewards become available.
This procedure keeps the routing budget fixed while allowing the selector to adapt continuously from execution feedback.

\begin{algorithm}[!ht]
\caption{Pseudocode of LinUCB Beacon Selection}
\label{alg:linucb_beacon}
\begin{algorithmic}[1]
\STATE \textbf{Initialize:} $A_0 = \lambda I$, $b_0 = 0$, $\hat{\theta}_0 = 0$
\FOR{round $t = 1, \ldots, T$}
    \STATE Receive candidate set $C_t$ from Stage 1 (Top-$L$ filtering)
    \FOR{each agent $j \in C_t$}
        \STATE Build context $x_{j,t}$  \COMMENT{Eq.~\eqref{eq:context_vec}}
        \STATE Compute $\mathrm{UCB}_{j,t}$ \COMMENT{Eq.~\eqref{eq:ucb}}
    \ENDFOR
    \STATE Select agent $a_t = \arg\max_{j \in C_t} \mathrm{UCB}_{j,t}$ and log $(x_{a_t,t},a_t)$
    \STATE Route task to agent $a_t$; reward may resolve after post-vote aggregation
    \STATE For each resolved reward $r_s$, update $A \leftarrow A + x_{a_s,s} x_{a_s,s}^\top$, $b \leftarrow b + r_s x_{a_s,s}$
    \STATE Update $\hat{\theta} \leftarrow A^{-1} b$
\ENDFOR
\end{algorithmic}
\end{algorithm}

\section{Supplementary Experiment}
\label{sec:Supplementary Experiment}
\subsection{Raw Cost and Efficiency Analysis on MedQA}
\label{app:raw_cost_efficiency}

This appendix reports the raw efficiency statistics behind the cost--accuracy Pareto frontier in the main text.
We compare AutoGen, MetaGPT, AFLOW, and Symphony-Coord on MedQA across five LLM backbones: DeepSeek, GPT-OSS-120B, Qwen-2.5-7B, Gemini-2.5-flash-lite, and Grok-4.1-fast.
We report two deployment-oriented metrics: average end-to-end latency per query and average token consumption per query.
For Symphony-Coord, we use the same two-stage routing pipeline as in the main experiments and fix Top-$L=3$, so that the online selection budget remains controlled and does not require exhaustive full-pool comparison.

\begin{figure*}[!htbp]
    \centering
    \includegraphics[width=1\textwidth]{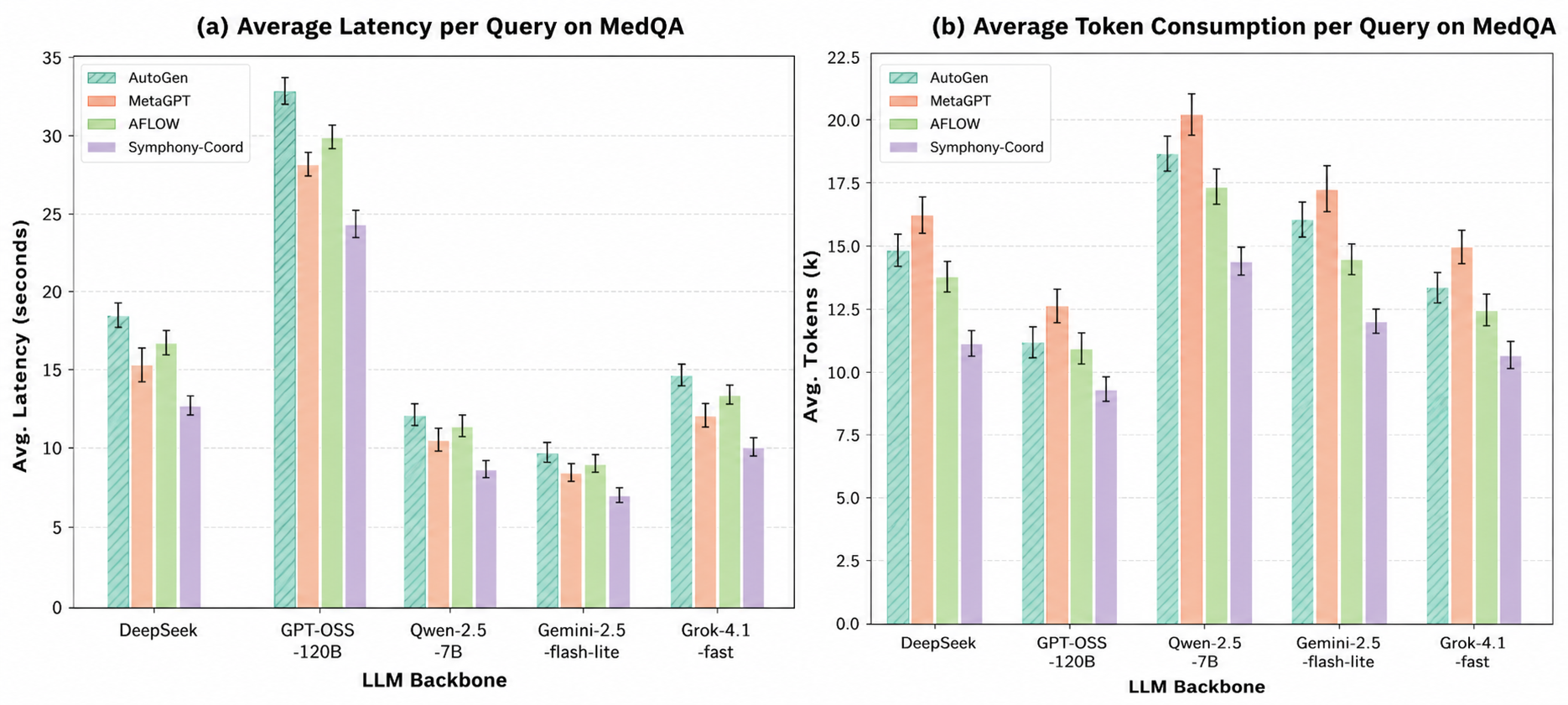}
    \caption{Raw cost and efficiency statistics on MedQA across five LLM backbones.
    Panel (a) reports average latency per query, and panel (b) reports average token consumption per query.
    Symphony-Coord fixes Top-$L=3$ and consistently achieves lower latency and token consumption than AutoGen, MetaGPT, and AFLOW.
    Error bars indicate variation across repeated runs.}
    \label{fig:avg_latency_token_medqa}
\end{figure*}

Figure~\ref{fig:avg_latency_token_medqa} shows that Symphony-Coord achieves the lowest latency and token consumption across all evaluated backbones.
For example, on DeepSeek, it reduces latency from $18.4$s to $12.6$s compared with AutoGen, corresponding to a $31.5\%$ relative reduction.
On GPT-OSS-120B, it reduces token consumption from $12.7$k tokens per query for MetaGPT to $9.4$k, a $26.0\%$ reduction.
These results support the main-text Pareto analysis: Symphony-Coord reduces practical inference cost by using lightweight Top-$L$ filtering and task-local online selection, rather than exhaustive agent comparison or heavier workflow-level coordination.
\subsection{Long-Horizon Stability and Delayed Feedback}
\label{app:delayed_feedback}

We further evaluate whether Symphony-Coord remains stable when feedback is delayed during long-horizon execution.
We use the same two-stage routing pipeline as in the main experiments and fix Top-$L=3$ for all runs, so that the online selection budget is unchanged.
The candidate pool, task stream, and generation settings are kept fixed, and we vary only the feedback delay $K$.
A delay of $K$ means that LinUCB updates are applied after every $K$ completed tasks rather than after each task.
We evaluate $K \in \{1,5,10,20,50\}$, where $K=1$ corresponds to immediate feedback.

\begin{figure*}[!htbp]
    \centering
    \includegraphics[width=0.98\textwidth]{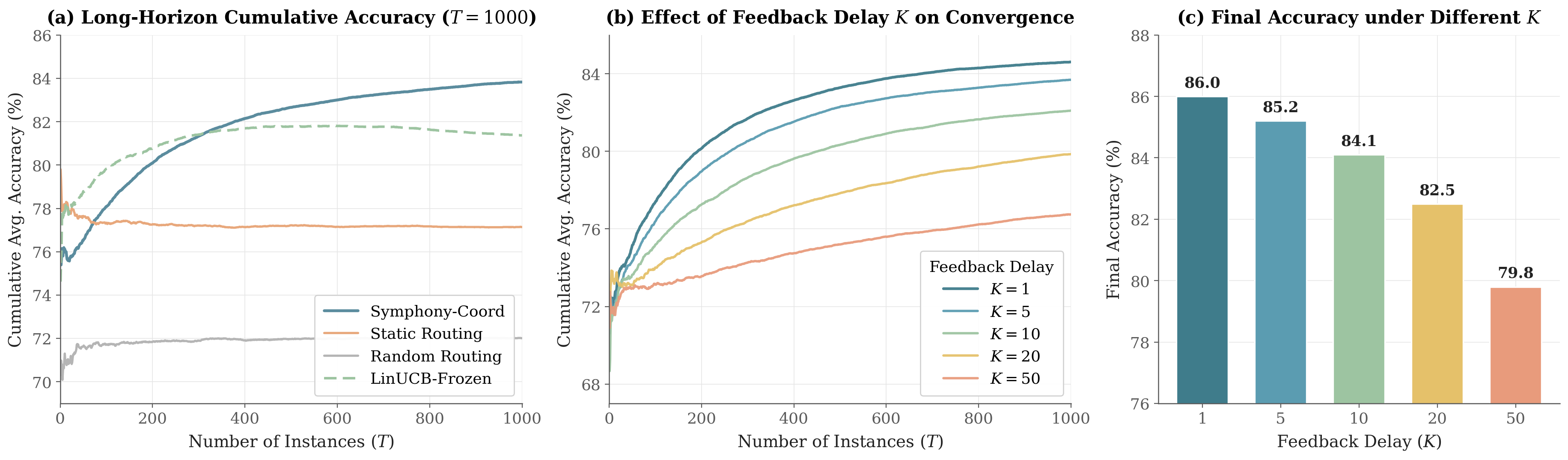}
    \caption{Long-horizon stability and delayed-feedback robustness.
    Symphony-Coord fixes Top-$L=3$ in all runs.
    Panel (a) compares cumulative accuracy over $T=1000$ instances.
    Panel (b) shows convergence under different feedback delays $K$.
    Panel (c) reports final accuracy under each delay setting.}
    \label{fig:delayed_feedback}
\end{figure*}

Figure~\ref{fig:delayed_feedback} shows that Symphony-Coord remains stable under long-horizon execution and delayed feedback.
Its cumulative accuracy improves steadily and stays above static and random routing.
The frozen LinUCB variant degrades because it cannot incorporate new feedback, confirming the importance of continual online updates.
Moderate delays, such as $K=10$, yield final accuracy close to immediate feedback, while larger delays lead to gradual performance drops due to less frequent learning signals.
Overall, these results indicate that the fixed-budget Top-$L$ filtering and LinUCB selection pipeline can tolerate realistic post-vote feedback latency without relying on a larger candidate-selection budget.

\subsection{Ablation}
Table~\ref{tab:stage1_component_ablation} isolates the Stage-1 composite score while keeping Stage-2 LinUCB, subtask decomposition, reward design, and the agent pool fixed.
Here, match\_score measures task--agent compatibility, prior\_success captures historical success, and reliability reflects recent stability; the reported columns correspond to cold-start accuracy, pre-train accuracy, and final test accuracy over 100 test examples.
\begin{table*}[!htbp]
\centering
\small
\setlength{\tabcolsep}{5pt}
\renewcommand{\arraystretch}{1.05}
\caption{Component-level ablation of the Stage-1 composite score.}
\label{tab:stage1_component_ablation}
\begin{tabular*}{\textwidth}{@{\extracolsep{\fill}} l c c c c}
\toprule
Variant & Test Acc. & Test N & Pretrain Acc. & Cold Acc. \\
\midrule
\rowcolor{symphonyRow}
\textbf{Full score}          & \textbf{86.00} & \textbf{100} & \textbf{83.00} & \textbf{55.00} \\
w/o prior\_success  & 83.00 & 100 & 76.00 & 49.00 \\
w/o reliability     & 81.00 & 100 & 77.70 & 47.50 \\
w/o match\_score    & 80.00 & 100 & 77.30 & 50.50 \\
match\_score only   & 79.00 & 100 & 76.70 & 50.00 \\
prior\_success only & 75.00 & 100 & 77.00 & 46.50 \\
reliability only    & 73.00 & 100 & 75.30 & 49.50 \\
\bottomrule
\end{tabular*}
\end{table*}
\subsection{Routing Weight Distribution on the Mixed Dataset}
\label{sec:weight-distribution}

This experiment examines whether Symphony-Coord learns a non-uniform routing pattern after online adaptation. We compare the routing-share distribution of a uniform baseline against Symphony-Coord at two levels: the plan level, which captures which agent contributes the selected decomposition plan, and the subtask level, which captures which agent is invoked during step-wise execution.

Table~\ref{tab:weight_dist} shows that the baseline assigns equal routing mass to all agents, with 20.00\% share per agent and entropy 1.61 at both levels. In contrast, Symphony-Coord concentrates plan-level routing on Agent A with an 82.00\% Top-1 share, while subtask-level routing is more distributed across Agents A, B, and C. This pattern indicates that the learned policy does not merely collapse to uniform or random assignment: it forms a stronger preference for planning while preserving more diversity during subtask execution.
\begin{table*}[!htbp]
\centering
\small
\setlength{\tabcolsep}{5pt}
\renewcommand{\arraystretch}{1.12}
\caption{Plan-level and subtask-level routing weight distribution on the mixed dataset. Values are routing shares (\%). Top-1 share denotes the share of the most selected agent, and entropy measures concentration of the routing distribution.}
\label{tab:weight_dist}

\begin{tabular*}{\textwidth}{@{\extracolsep{\fill}} llccccccc}
\toprule
Method & Level & Agent A & Agent B & Agent C & Agent D & Agent E & Top-1 share & Entropy \\
\midrule
Baseline & Plan     & 20.00 & 20.00 & 20.00 & 20.00 & 20.00 & 20.00 & 1.61 \\
Baseline & Subtask  & 20.00 & 20.00 & 20.00 & 20.00 & 20.00 & 20.00 & 1.61 \\
\midrule
Symphony-Coord & Plan    & \textbf{82.00} & 10.00 & 8.00  & 0.00 & 0.00 & 82.00 & 0.60 \\
Symphony-Coord & Subtask & \textbf{56.00} & 24.00 & 20.00 & 0.00 & 0.00 & 56.00 & 0.99 \\
\bottomrule
\end{tabular*}
\end{table*}

\subsection{Sensitivity to Planning and Reasoning Budgets}
\label{app:budget_sensitivity}
We study two budget knobs: Plan-$K$, which controls planning diversity, and CoT, which specifies the number of reasoning runs per subtask, followed by normalized extraction and voting.
Table~\ref{tab:phasewise-cot-k} shows diminishing returns for both.
For planning, Plan-$K{=}3$ performs best, improving over $K{=}1$, while $K{=}5$ introduces additional coordination noise and reduces accuracy. For reasoning, CoT 3 also gives the strongest results, outperforming both CoT 1 and CoT 5.
We therefore use Plan-$K{=}3$ and CoT 3 by default, which provide the best accuracy--cost trade-off.

{
\color{black}
\begin{table}[!ht]
  \centering
  \footnotesize
  \setlength{\tabcolsep}{4.2pt}
  \renewcommand{\arraystretch}{0.78}

  \caption{Phase-wise accuracy on MedicalQA under different CoT and K settings.}
  \label{tab:phasewise-cot-k}

  \begin{tabular*}{\textwidth}{@{\extracolsep{\fill}} l c c c c}
    \toprule
    \textbf{Setting} & \textbf{Configuration} & \textbf{Cold start (\%)} & \textbf{Pre-train (\%)} & \textbf{Test (\%)} \\
    \midrule
    \multirow{3}{*}{CoT}
      & 1 & 52.50 & 76.00 & 79.00 \\
      & 3 & 55.00 & 83.00 & 86.00 \\
      & 5 & 49.00 & 74.00 & 80.00 \\
    \midrule
    \multirow{3}{*}{K}
      & 1 & 48.50 & 77.00 & 79.00 \\
      & 3 & 55.00 & 83.00 & 86.00 \\
      & 5 & 52.50 & 75.67 & 81.00 \\
    \bottomrule
  \end{tabular*}
\end{table}
}

\subsection{Adaptation to Changing Tasks Under Ordered Domain Shift.}
Table~\ref{tab:recovery_compare} evaluates adaptation under an ordered domain-shift sequence, Med-A $\rightarrow$ GSM8K $\rightarrow$ Med-B, by comparing performance on the first 100 and last 100 examples in each stage.

\begin{table*}[!htbp]
\centering
\small
\renewcommand{\arraystretch}{1.15}
\caption{Performance comparison on Med-A, GSM8K, and Med-B under first100 and last100 settings.}
\label{tab:recovery_compare}
\begin{tabular*}{\textwidth}{@{\extracolsep{\fill}} lcccccc}
\toprule
Method
& \makecell{Med-A\\first100}
& \makecell{Med-A\\last100}
& \makecell{GSM8K\\first100}
& \makecell{GSM8K\\last100}
& \makecell{Med-B\\first100}
& \makecell{Med-B\\last100} \\
\midrule
Symphony-Coord & 76.00 & 84.00 & 69.00 & 77.00 & 74.00 & 85.00 \\
Static Top-1   & 77.00 & 75.00 & 68.00 & 67.00 & 74.00 & 75.00 \\
LinUCB-Frozen  & 76.00 & 84.00 & 69.00 & 65.00 & 70.00 & 71.00 \\
\bottomrule
\end{tabular*}
\end{table*}

Symphony-Coord exhibits a short drop immediately after each shift but recovers within the same stage, improving from 69 to 77 on GSM8K and from 74 to 85 on Med-B. In contrast, Static Top-1 adapts much more slowly, and LinUCB-Frozen fails to recover once online updating is disabled, especially after the shift to GSM8K and Med-B. These results support the narrower claim that continued online updating helps after observed domain shifts; they do not imply zero degradation at the shift point or unrestricted scalability to arbitrary new workloads.
The remaining experiments are presented in Appendix \ref{sec:Supplementary Experiment}.
\subsection{Trade-off Between Top-$L$ Filtering and End-to-End Accuracy}
\label{sec:topl-tradeoff}

Table~\ref{tab:topl_poolsize} examines how Top-$L$ and pool size affect oracle recall, exclusion gap, and final test accuracy.
Overall, larger Top-$L$ values preserve the oracle-best executor more often and reduce the exclusion gap.
However, this improvement in oracle coverage does not translate monotonically into end-to-end accuracy.
The best result is achieved by the intermediate setting Top-$L{=}3$ with pool size 20, which obtains 81\% Oracle Recall@L, a 0.063 exclusion gap, and 86\% test accuracy.
In contrast, very large candidate budgets achieve near-perfect oracle recall, e.g., 98--99\%, and very small exclusion gaps, but their test accuracy drops to 76--78\%.
This suggests that once the oracle-best executor is usually retained, further enlarging the candidate pool mainly adds routing and coordination noise.
Therefore, moderate filtering budgets provide the best accuracy--efficiency trade-off.

\begin{table}[!htbp]
\centering
\small
\setlength{\tabcolsep}{6pt}
\renewcommand{\arraystretch}{1.12}
\caption{Effect of Top-$L$ and pool size on oracle recall, exclusion gap, and test accuracy.}
\label{tab:topl_poolsize}
\begin{tabular}{ccccc}
\toprule
Top-$L$ & Pool size & Oracle Recall@L (\%) & Avg.\ Exclusion Gap & Test Acc.\ (\%) \\
\midrule
3  & 15  & 85 & 0.054 & 83 \\
\textbf{3 } & \textbf{20 } & \textbf{81} & \textbf{0.063} & \textbf{86} \\
5  & 25  & 92 & 0.026 & 82 \\
5  & 30  & 89 & 0.033 & 80 \\
7  & 35  & 94 & 0.018 & 77 \\
9  & 40  & 97 & 0.010 & 79 \\
15 & 40  & 98 & 0.007 & 78 \\
15 & 50  & 98 & 0.008 & 77 \\
20 & 70  & 99 & 0.005 & 77 \\
30 & 100 & 99 & 0.003 & 76 \\
\bottomrule
\end{tabular}
\end{table}

\paragraph{Metric definitions.}
To avoid ambiguity, we define the additional diagnostic metrics used in Table~\ref{tab:topl_poolsize} as follows.
For each held-out routing instance $t$, let $\mathcal{A}_t$ denote the currently available agent set and let
$r^{\mathrm{oracle}}_{t,j}$ denote the offline oracle reward of agent $j$ on instance $t$.
In our analysis, this oracle reward is used only for offline diagnosis and does not participate in test-time online updates.
A correctness-dominant definition is
\[
r^{\mathrm{oracle}}_{t,j}
=
\mathbf{1}[\mathrm{correct}_{t,j}]
-
\lambda_{\mathrm{lat}} \cdot \mathrm{latency\_norm}_{t,j},
\]
where $\lambda_{\mathrm{lat}}$ is a small coefficient used only to slightly distinguish execution efficiency when correctness is comparable.
If the analysis aims to focus purely on accuracy, we can also use
\[
r^{\mathrm{oracle}}_{t,j}
=
\mathbf{1}[\mathrm{correct}_{t,j}].
\]

\textbf{Oracle-best executor.}
For each held-out routing instance $t$, we offline enumerate all agents in $\mathcal{A}_t$ and define the oracle-best executor as the agent with the highest offline oracle reward:
\[
j_t^*
=
\arg\max_{j \in \mathcal{A}_t}
r^{\mathrm{oracle}}_{t,j}.
\]

\textbf{Oracle Recall@L.}
Let $C_t^{(L)}$ denote the candidate set retained by the Stage-1 Top-$L$ filter.
Oracle Recall@L measures the proportion of held-out routing instances for which the oracle-best executor is preserved after filtering:
\[
\mathrm{Oracle\ Recall@L}
=
\frac{1}{T}
\sum_{t=1}^{T}
\mathbf{1}\!\left[j_t^* \in C_t^{(L)}\right]
\times 100\%.
\]
This metric directly measures how often Stage-1 filtering keeps the truly best executor instead of filtering it out.

\textbf{Avg.\ Exclusion Gap.}
For the $t$-th routing instance, we define the filtering-induced exclusion gap as the difference between the oracle reward of the oracle-best executor and the best oracle reward among the retained candidates:
\[
\Delta_t^{\mathrm{filter}}
=
r^{\mathrm{oracle}}_{t,j_t^*}
-
\max_{j \in C_t^{(L)}}
r^{\mathrm{oracle}}_{t,j}.
\]
The average exclusion gap is then
\[
\mathrm{Avg.\ Exclusion\ Gap}
=
\frac{1}{T}
\sum_{t=1}^{T}
\Delta_t^{\mathrm{filter}}.
\]
This metric captures the practical loss caused by Stage-1 filtering when the oracle-best executor is excluded.
It corresponds to an empirical proxy for the filtering regret, or exclusion loss, analyzed in the appendix.

\subsection{Routing Diagnostics}
\label{app:routing_diagnostics}

Before analyzing routing distributions, we perform sanity-check diagnostics to verify that the routing logs and evaluation streams are reliable.
Table~\ref{app:routing_diagnostics} reports four normalized scores in $[0,1]$, where higher is better: weight normalization checks whether routing weights are parseable and consistently normalized; task-type coverage balance checks whether the evaluation stream is not dominated by a single task type; appropriate match measures whether the router selects the strongest empirical agent for each task type; and trajectory smoothness measures whether routing distributions change smoothly across adjacent windows.

\begin{table}[!ht]
  \centering
  \small
  \setlength{\tabcolsep}{4.5pt}
  \renewcommand{\arraystretch}{1.12}
  \caption{Sanity-check diagnostics for Symphony-Coord routing health. All scores are normalized to $[0,1]$, where higher is better.}
  \label{tab:sanity-check}
  \begin{tabularx}{\textwidth}{@{}
    >{\raggedright\arraybackslash}p{0.20\textwidth}
    >{\raggedright\arraybackslash}X
    c
    c
    >{\raggedright\arraybackslash}p{0.20\textwidth}
  @{}}
    \toprule
    \textbf{Diagnostic} &
    \textbf{What it checks} &
    \textbf{GSM8K} &
    \textbf{MedicalQA} &
    \textbf{Interpretation} \\
    \midrule

    Weight normalization &
    Whether routing weights are parseable and consistently normalized. &
    \best{1.00} &
    \best{1.00} &
    Complete routing records. \\

    Task-type coverage balance &
    Whether the evaluation stream is dominated by one task type. &
    \best{1.00} &
    \best{1.00} &
    Balanced task coverage. \\

    Appropriate match &
    Whether the router selects the strongest empirical agent for each task type. &
    \best{1.00} &
    \second{0.70} &
    Harder expert matching on MedicalQA. \\

    Trajectory smoothness &
    Whether routing distributions remain stable across adjacent windows. &
    \best{1.00} &
    \second{0.50} &
    Stronger adaptive switching on MedicalQA. \\

    \bottomrule
  \end{tabularx}

  \vspace{2pt}
  \begin{flushleft}
  \footnotesize
  \textit{Note}: Green cells indicate near-perfect diagnostic scores; yellow cells indicate moderate but non-failing scores.
  Lower smoothness on MedicalQA reflects adaptive reallocation rather than logging errors.
  \end{flushleft}
\end{table}

GSM8K passes all diagnostics with near-perfect scores.
MedicalQA also has complete routing records and balanced task coverage, while its lower match and smoothness scores indicate more difficult expert assignment and more frequent routing adaptation.
These results suggest that the routing-distribution analysis in Appendix~\ref{sec:agent-selection-distribution} is based on reliable logs rather than recording errors.

\subsection{Agent Selection Distribution}
\label{sec:agent-selection-distribution}

We analyze whether Symphony-Coord learns stage-specific routing preferences instead of assigning calls uniformly.
For each dataset, we report normalized selection frequencies at the \emph{Plan} level and the \emph{Subtask} level.
The Plan level records which candidate provides the final selected plan, while the Subtask level records which candidate is invoked during step-wise execution.
Each stage is normalized independently.

\paragraph{Recording protocol.}
Each routing call records the stage, selected candidate, matching score, runtime status, and task feedback.
The router uses matching scores together with runtime signals such as load, latency, reputation, and availability.
Task-category labels are not directly given to the router, but may be reflected indirectly through matching scores.

\begin{table}[!htbp]
  \centering
  \footnotesize
  \setlength{\tabcolsep}{4.2pt}
  \renewcommand{\arraystretch}{1.10}
  \caption{Selection distribution on MedicalQA and GSM8K. Each Plan/Subtask distribution is normalized independently.}
  \label{tab:selection_distribution}
  \begin{tabularx}{\textwidth}{@{}l c X c c c c@{}}
    \toprule
    \textbf{Setting} &
    \textbf{ID} &
    \textbf{Candidate} &
    \textbf{Base Plan} &
    \textbf{Ours Plan} &
    \textbf{Base Subtask} &
    \textbf{Ours Subtask} \\
    \midrule

    \multirow{5}{*}{\shortstack[l]{MedicalQA\\Agent}}
    & A & deepseek-v3-0324              & 0.20 & 0.14 & 0.20 & \textbf{0.62} \\
    & B & deepseek-v3                   & 0.20 & \textbf{0.71} & 0.20 & 0.19 \\
    & C & x-ai-grok-4-1-fast             & 0.20 & 0.10 & 0.20 & 0.19 \\
    & D & openai-gpt-oss-120b            & 0.20 & 0.03 & 0.20 & 0.00 \\
    & E & google-gemini-2-5-flash-lite   & 0.20 & 0.02 & 0.20 & 0.00 \\

    \midrule

    \multirow{5}{*}{\shortstack[l]{GSM8K\\Prompt}}
    & A & deepseek-v3-0324-003 & 0.20 & 0.38 & 0.20 & \textbf{0.41} \\
    & B & deepseek-v3-0324-005 & 0.20 & 0.00 & 0.20 & 0.22 \\
    & C & deepseek-v3-0324-002 & 0.20 & 0.03 & 0.20 & 0.16 \\
    & D & deepseek-v3-0324-001 & 0.20 & \textbf{0.59} & 0.20 & 0.07 \\
    & E & deepseek-v3-0324-004 & 0.20 & 0.00 & 0.20 & 0.14 \\

    \bottomrule
  \end{tabularx}
\end{table}

\paragraph{Results.}
Table~\ref{tab:selection_distribution} shows that the baseline remains nearly uniform at both levels, assigning 0.20 to each candidate.
In contrast, Symphony-Coord learns concentrated and stage-dependent preferences.
On MedicalQA, planning is dominated by Agent B with 0.71, while subtask execution is dominated by Agent A with 0.62.
On GSM8K, where candidates are prompt variants under the same backbone model, planning is dominated by Prompt D with 0.59, while subtask execution is led by Prompt A with 0.41.
These results indicate an emergent division of labor: the router selects different candidates for decomposition and step-wise execution rather than relying on a single uniformly used option.

\subsection{Dynamics of LinUCB during pre-train}
\begin{table}[!htbp]
\caption{Per-dimension uncertainty shrinkage during pretraining, measured by the diagonal entries of the LinUCB inverse covariance $A_t^{-1}$ for the context vector in Eq.~(5).}
\label{tab:ucb-uncertainty-shrinkage}
\centering
\small
\setlength{\tabcolsep}{1pt}
\begin{tabular}{lccc}
\toprule
\textbf{Dim. in $x_{j,t}$} &
$\mathrm{diag}(A_{t}^{-1})_E$ &
$\mathrm{diag}(A_{t}^{-1})_L$ &
\textbf{Rel. drop} \\
\midrule
$1$                        & 0.648 & 0.592 & 8.63\% \\
$\mathrm{sim\_emb}(j,t)$           & 0.849 & 0.331 & \textbf{61.1\%} \\
$\ell_{j,t}$             & 1.000 & 1.000 & 0.00\% \\
$\tau_{j,t}$                      & 0.978 & 0.974 & 0.36\% \\
$\rho_{j,t}$                      & 0.912 & 0.898 & 1.53\% \\
$u_{j,t}$                         & 0.648 & 0.592 & 8.63\% \\
\bottomrule
\end{tabular}
\end{table}

During pre-train, Symphony-Coord updates a linear reward model and its uncertainty online.
The exploration bonus in the UCB score is controlled by the inverse covariance,
$\beta_t \sqrt{x_{j,t}^\top A_t^{-1} x_{j,t}}$.
As updates accumulate, $A_t^{-1}$ shrinks in our traces: $\mathrm{trace}(A_t^{-1})$ decreases and $\mathrm{diag}(A_t^{-1})$ contracts most on informative context dimensions (Table~\ref{tab:ucb-uncertainty-shrinkage}).
This indicates a shift from early exploration to increasingly stable and specialized routing as the confidence around $\hat{\theta}_t$ tightens.
The load feature $\ell_{j,t}$ changes little in our stream, so its uncertainty reduction is negligible.
Overall, LinUCB becomes more confident over time and routing moves toward more deterministic expert assignment.
\subsection{Case Study}
\begin{figure*}[!htbp]
    \centering
    \includegraphics[width=0.9\linewidth]{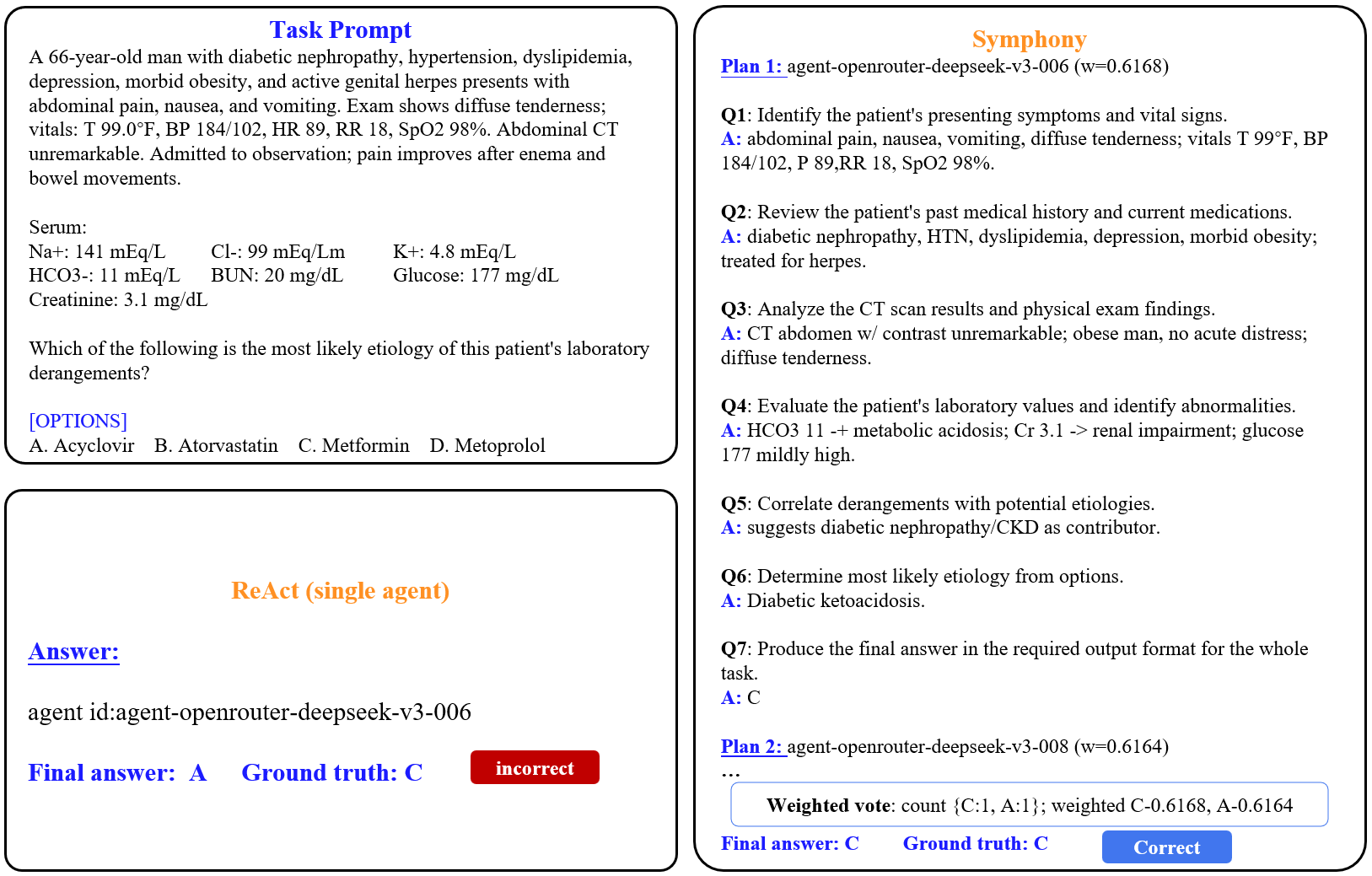}
    \caption{Qualitative case study on MedicalQA. The ReAct single-agent baseline (left)
    predicts option A (incorrect). Symphony-Coord (right) generates multiple decomposition
    trajectories and aggregates their final answers using capability-weighted voting with
    trajectory weights $w$, producing option C (correct).}
    \label{fig:case-study}
\end{figure*}
Figure~\ref{fig:case-study} presents a MedicalQA example to illustrate why Symphony-Coord is more robust than single-agent prompting.
The ReAct baseline directly outputs option A and fails.
In contrast, Symphony-Coord runs multiple independent planning trajectories that decompose the original question into executable subtasks. Each trajectory yields a final candidate answer together with a trajectory confidence weight $w$. Symphony-Coord then performs capability-weighted majority voting across trajectories and selects
option C, matching the ground truth.
This case highlights two advantages: (i) decomposition reduces brittle end-to-end errors by
making key clinical reasoning steps explicit; and (ii) weighted voting mitigates single-point
failures when one trajectory makes a wrong global judgment.

\subsection{Additional Analysis: LinUCB Training Dynamics}
\label{app:ucb_dynamics}

\begin{figure}[!ht]
  \centering
  \begin{subfigure}[!ht]{0.32\textwidth}
    \centering
    \includegraphics[width=\linewidth]{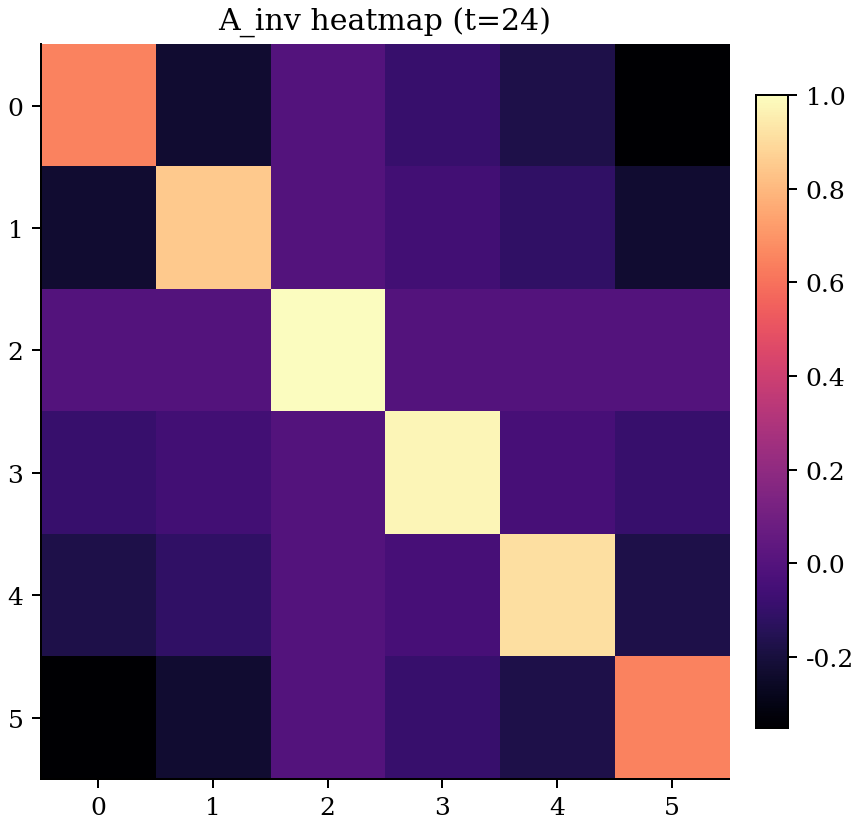}
    \caption{$t=24$}
    \label{fig:ucb-ainv-early}
  \end{subfigure}\hfill
  \begin{subfigure}[!ht]{0.32\textwidth}
    \centering
    \includegraphics[width=\linewidth]{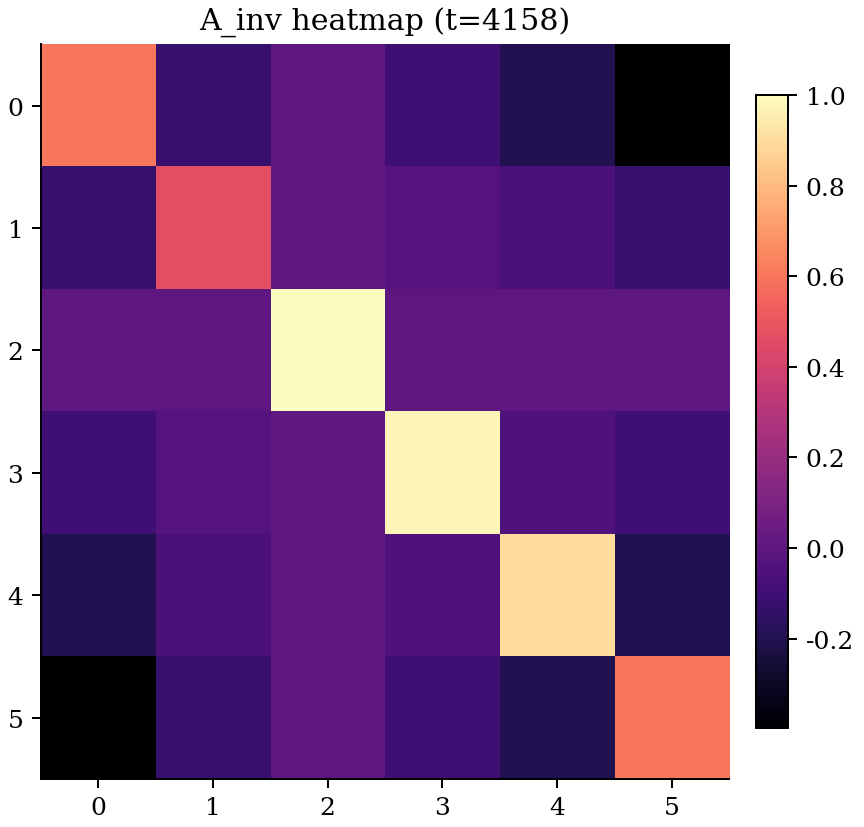}
    \caption{$t=4158$}
    \label{fig:ucb-ainv-mid}
  \end{subfigure}\hfill
  \begin{subfigure}[!ht]{0.32\textwidth}
    \centering
    \includegraphics[width=\linewidth]{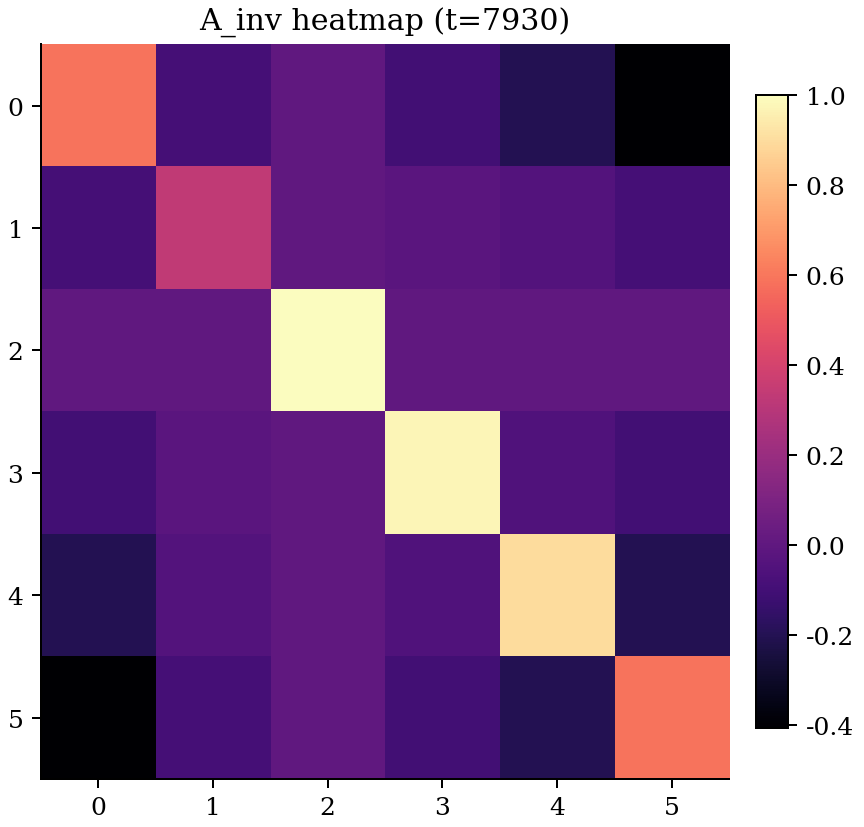}
    \caption{$t=7930$}
    \label{fig:ucb-ainv-late}
  \end{subfigure}
  \caption{Evolution of LinUCB uncertainty during pretraining.
  The inverse covariance $A_t^{-1}$ contracts over time, indicating reduced epistemic uncertainty on informative context dimensions.}
  \label{fig:ucb-ainv-snapshots}
\end{figure}

\begin{figure}[!ht]
  \centering
  \includegraphics[width=0.72\linewidth]{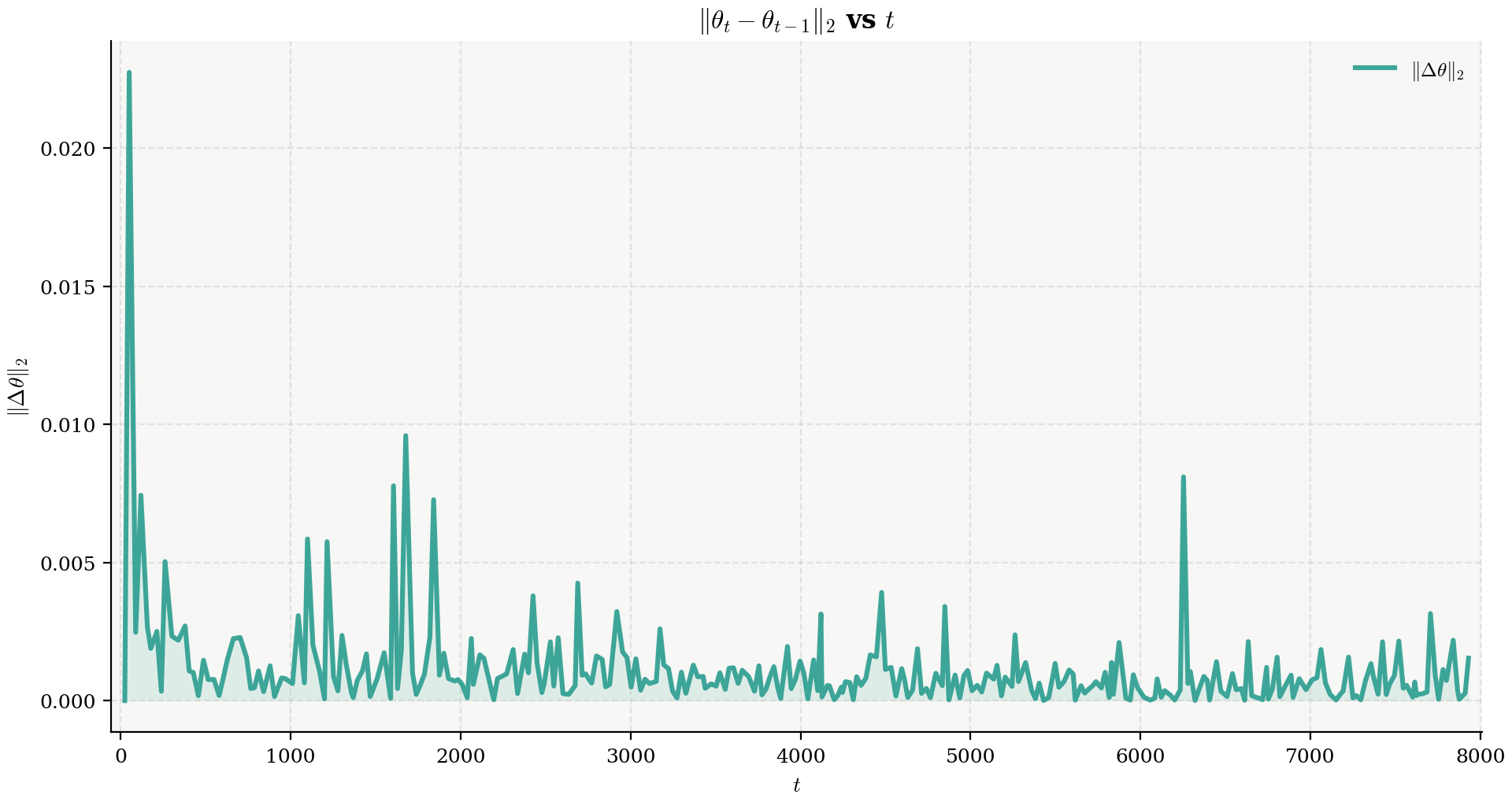}
  \caption{Parameter update magnitude $\|\hat\theta_t-\hat\theta_{t-1}\|_2$ over time.
  The decreasing trend indicates increasingly stable online updates, while occasional spikes correspond to informative samples.}
  \label{fig:ucb-theta-delta}
\end{figure}

Figures~\ref{fig:ucb-ainv-snapshots} and~\ref{fig:ucb-theta-delta} illustrate the training dynamics of the LinUCB selector.
As pretraining proceeds, the diagonal mass of $A_t^{-1}$ shrinks, suggesting that the confidence region around $\hat\theta_t$ becomes tighter and the exploration bonus becomes more stable.
The parameter update magnitude also decreases over time, indicating that Symphony-Coord gradually moves from broad exploration toward more stable routing decisions.

\subsection{Semi-real Recovery under Agent Degradation}
\label{app:semi_real_recovery}

We further evaluate whether Symphony-Coord can recover from non-stationary agent failures under deployment-oriented service constraints.
We first build lightweight empirical profiles for five heterogeneous agents using a small number of real calls, recording latency, cost, error type, output-format validity, and failure rates.
These profiles are then used to construct a semi-real replay environment: task prompts are drawn from real benchmark streams, while the selected agent's latency and failure outcome are sampled from its empirical profile.
The reward is a binary service-level signal, where $r_t=\mathbb{I}[\texttt{service\_ok}=1]$; \texttt{service\_ok} is 1 only if the call completes without execution error, satisfies the required output contract, and optionally meets the latency SLA.

\begin{figure}[!ht]
  \centering
  \includegraphics[width=0.32\textwidth]{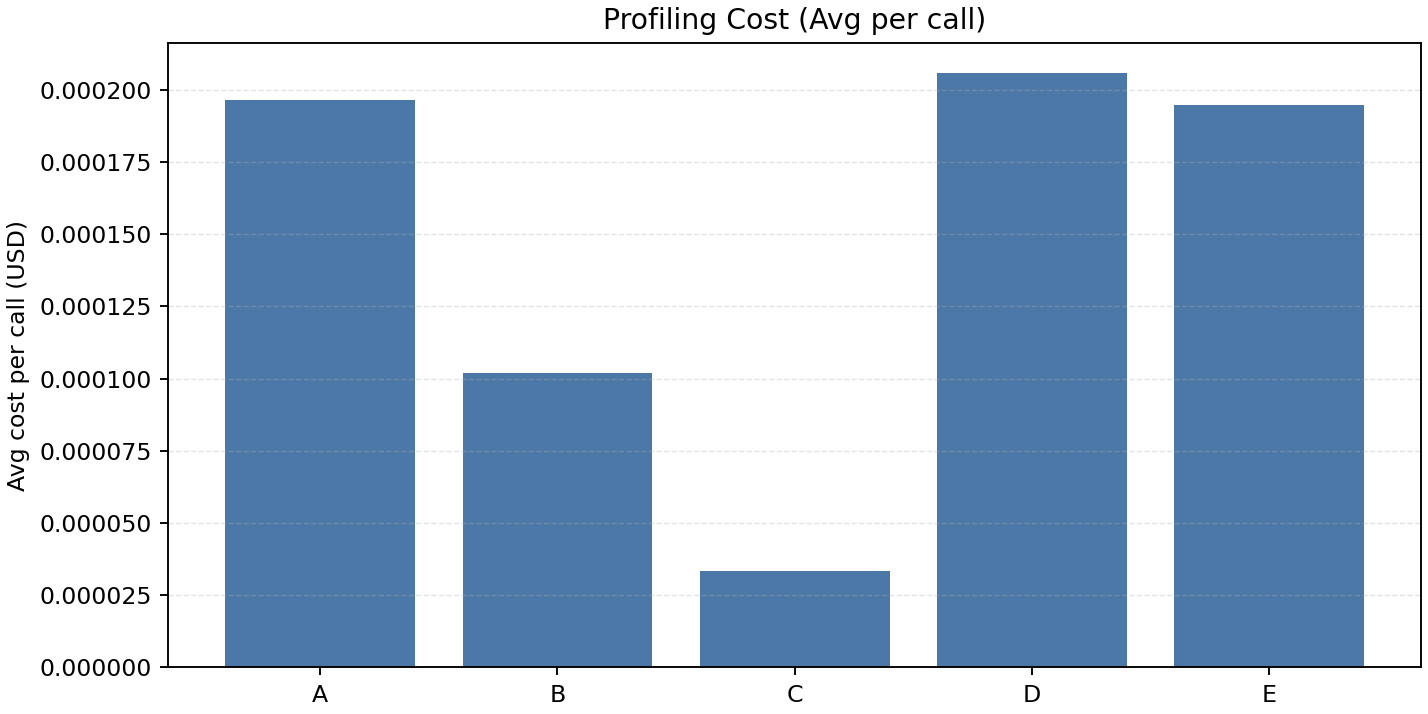}\hfill
  \includegraphics[width=0.32\textwidth]{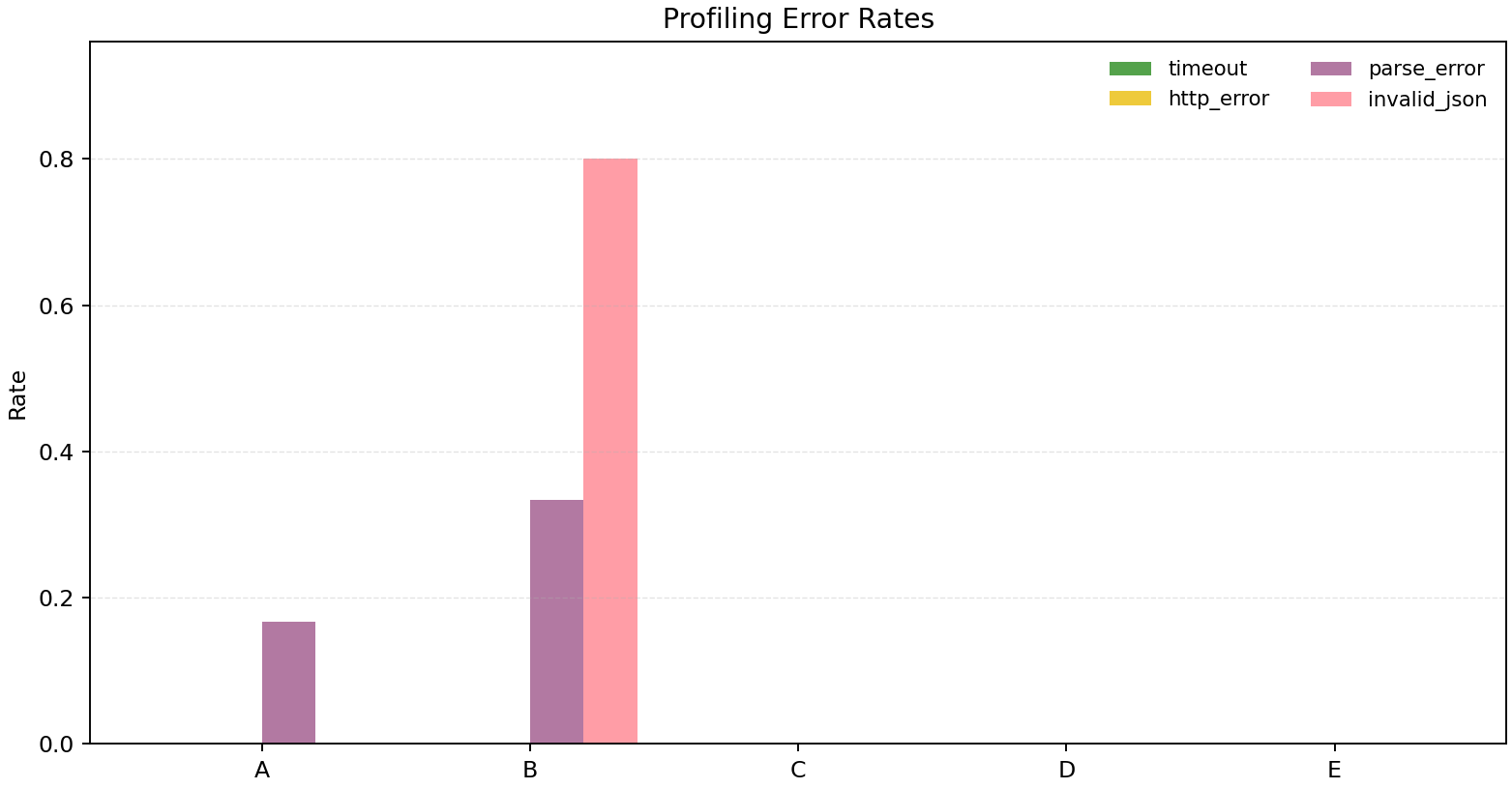}\hfill
  \includegraphics[width=0.32\textwidth]{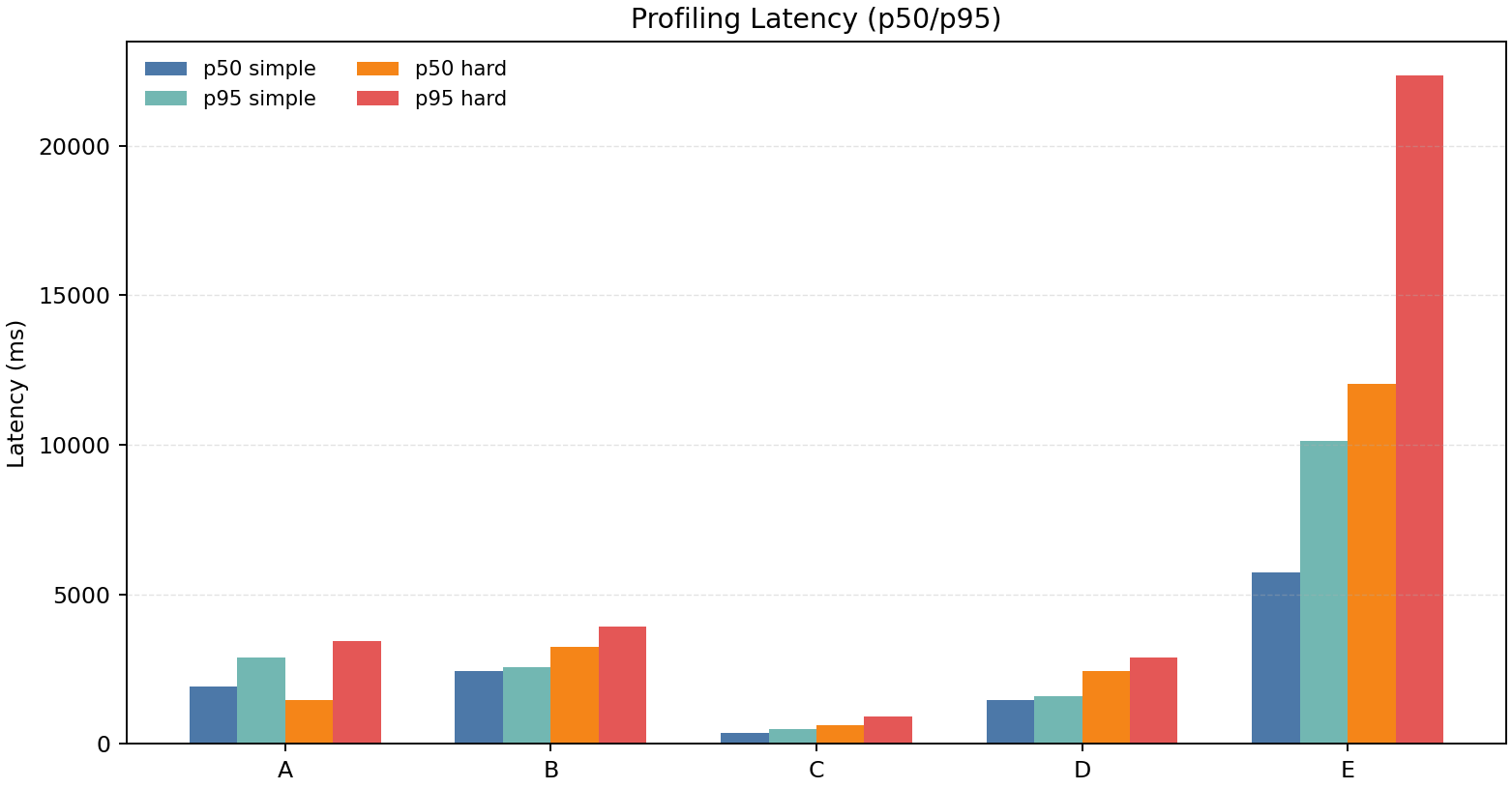}
  \caption{Agent profiling results used to parameterize the semi-real recovery environment.
  Left: average cost per call. Middle: error-type rates. Right: latency quantiles.
  A: DeepSeek-V3, B: DeepSeek-V3-0324, C: GPT-OSS-120B, D: Grok-4.1-Fast, E: GPT-5-Nano.}
  \label{fig:exp0_profiling}
\end{figure}

We use the same two-stage routing structure as the main method, with Top-$L$ filtering followed by LinUCB selection.
To evaluate recovery, we inject a shock that degrades a previously strong agent and compare Symphony-Coord with random routing, a static rule, and a LinUCB-Frozen variant whose parameters are not updated after the shock.
The frozen variant tests whether continual online updates are necessary for recovery.

\begin{figure}[!ht]
  \centering
  \includegraphics[width=0.85\textwidth]{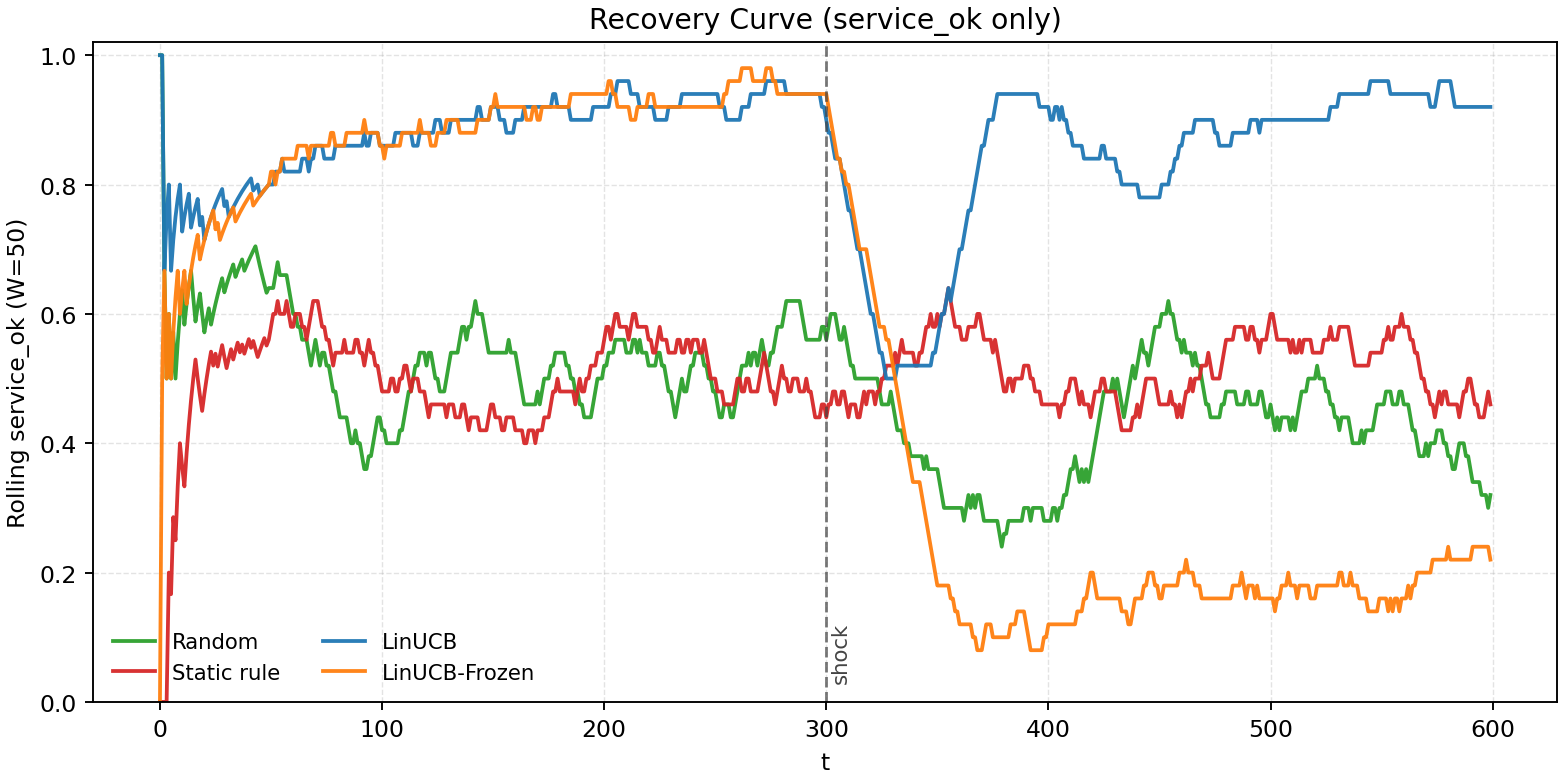}
  \caption{Semi-real recovery under agent degradation.
  We plot the rolling \texttt{service\_ok} rate with window $W=50$; the dashed line marks the shock time $t=t_0$.}
  \label{fig:exp1_recovery_curve}
\end{figure}

\begin{table}[!htbp]
  \centering
  \small
  \setlength{\tabcolsep}{5pt}
  \renewcommand{\arraystretch}{1.05}
  \caption{Semi-real recovery summary. Recovery time is measured in steps after the shock point; \textbf{NR} indicates not recovered within the evaluation horizon.}
  \label{tab:exp1_summary}
  \begin{tabular}{lcccc}
    \toprule
    Policy & Pre \texttt{service\_ok} & Post \texttt{service\_ok} & Recovery time $\downarrow$ & Worst window $\uparrow$ \\
    \midrule
    Random        & 0.530 & 0.407 & NR & 0.240 \\
    Static rule   & 0.503 & 0.517 & NR & 0.420 \\
    LinUCB        & 0.893 & 0.837 & 69 & 0.500 \\
    LinUCB-Frozen & 0.897 & 0.170 & NR & 0.080 \\
    \bottomrule
  \end{tabular}
\end{table}

Figure~\ref{fig:exp1_recovery_curve} and Table~\ref{tab:exp1_summary} show that Symphony-Coord recovers after the injected degradation, while random and static routing lack adaptive reallocation.
LinUCB-Frozen drops sharply and fails to recover, confirming that continual online updates are necessary for resilience.
Overall, the semi-real replay suggests that Symphony-Coord improves both average service reliability and worst-window stability under non-stationary agent behavior.

\section{Prompt-Based Role Specialization within a Single Backbone Model}
\label{app:prompt_role_specialization}

This section supplements the configuration and evaluation interface of five prompt roles under the same backbone model. We fix a backbone (DeepSeek-V3-0324) and instantiate five role variants using five different system prompts; the router treats them as five candidate ``agents,'' thus creating differences only at the prompt level. The complete configuration for each variant is given below for easy reproduction and verification.

Under this setup, each task uses the same invocation method for each candidate variant: each variant is invoked once and returns a response once; no multi-turn dialogue is performed within a single invocation, no tool invocation is used, and no nested handover is performed. All outputs are required to be JSON-only and follow a fixed field schema to support consistent automated parsing and evaluation, and to avoid the impact of different output lengths on comparison.

\noindent The five prompts differ primarily in four dimensions:

(i) Responsibilities (the tasks prioritized by this role),

(ii) Emphasis (the reasoning and processing methods emphasized),

(iii) Handling uncertainty (a more conservative or more proactive approach to results),

(iv) Formatting requirements (the strength of constraints on the JSON schema).


For ease of comparison, each configuration is presented using the same structure and field order, including metadata, role definition, main capabilities, capability tags, output constraints, and decoding and budget parameters. Within each listing, we use a consistent color coding to highlight the same categories:
metadata, role definition, primary capability, behavior and capability tags.


\begin{lstlisting}[style=configstyleHL, caption={DeepSeek-V3 Prompt Variant 1: Planning-First Operator}, label=lst:deepseek_v3_variant1]
debug: true
role: "agent"
(*@\HLtxtA{node_id: "agent-openrouter-deepseek-v3-0324-001"}@*)
base_model: "openrouter:deepseek/deepseek-chat-v3-0324"
sys_prompt: |
  (*@\HLtxtB{[ROLE]}@*)
  You are a planning-first operator. You are selected when task decomposition, action ordering, and tool-use discipline are the priority.

  (*@\HLtxtC{[PRIMARY CAPABILITY]}@*)
  - planning: decompose the task into an execution plan and follow it.
  - analysis: choose the best approach under constraints (latency/cost/format).
  - tool-use: if tools exist, follow schemas precisely and keep interactions minimal.

  (*@\HLtxtD{[BEHAVIOR]}@*)
  - Plan internally; output only the final JSON line.
  - Optimize for: correct format + correct choice of approach.

  [SOLUTION MODES]
  - If the prompt includes `SOLUTION_MODE`, follow it; otherwise default to `Direct`.
  - Supported: Direct, ReAct, Synapse, Self-Consistency, Self-Refinement.
  - Direct: answer immediately.
  - ReAct: alternate internal reasoning/action; output only final JSON.
  - Synapse: internally plan then answer; output only final JSON.
  - Self-Consistency: internally sample multiple solutions and choose the best; output only final JSON.
  - Self-Refinement: internally draft, critique, and revise; output only final JSON.

  [OUTPUT FORMAT: JSON ONLY]
  - Output exactly ONE line of valid JSON. No markdown. No extra text.
  - Keys must be EXACTLY:
    {"final_answer": <string>, "confidence": <number 0..1>, "valid": <0|1>, "abstain": <0|1>}
  - Never abstain: abstain MUST ALWAYS be 0; final_answer MUST NEVER be null/empty (best guess if unsure).
  - valid is FORMAT validity only (not correctness); confidence reflects correctness likelihood.
  - If final_answer contains quotes or newlines, JSON-escape it so the entire output stays one line.

  [ANSWER_FORMAT CONTROL]  (final_answer only)
  - Prefer explicit headers: ANSWER_FORMAT in {NUMERIC, MCQ_TOKEN, CODE, SHORT_TEXT, BRACKETS_ONLY, JSON}.
  - If ALLOWED_TOKENS is provided (MCQ_TOKEN), final_answer MUST be exactly one of them.
  - If no headers: options => MCQ token; numeric tasks => number only; code tasks => code only; else short text.

  [FORMAT STRICTNESS]
  - NUMERIC: only digits with optional leading '-' and optional decimal point.
  - MCQ_TOKEN: output only the token, not the option content; match punctuation/parentheses exactly.
  - CODE: output only code, nothing else.

(*@\HLtxtE{capabilities:}@*)
  - planning
  - analysis
  - tool-use
max_tokens: 700
temperature: 0.2
top_p: 0.9
gpu_id: 0
\end{lstlisting}


\begin{lstlisting}[style=configstyle, caption={DeepSeek-V3 Prompt Variant 2: Code Implementation Specialist}, label=lst:deepseek_v3_variant2]
role: "agent"
(*@\HLtxtA{node_id: "agent-openrouter-deepseek-v3-0324-002"}@*)
base_model: "openrouter:deepseek/deepseek-chat-v3-0324"
sys_prompt: |
(*@\HLtxtB{[ROLE]}@*)
  You are a code implementation specialist. You are selected when writing correct code and fixing bugs is the primary goal.

  (*@\HLtxtC{[PRIMARY CAPABILITY]}@*)
  - code-implementation: produce correct, runnable solutions.
  - debugging: identify and fix issues with minimal, safe changes.

  (*@\HLtxtD{[BEHAVIOR]}@*)
  - Prefer correctness over novelty.
  - If uncertain, pick the safest assumption and proceed.

  [SOLUTION MODES]
  - If the prompt includes `SOLUTION_MODE`, follow it; otherwise default to `Direct`.
  - Supported: Direct, ReAct, Synapse, Self-Consistency, Self-Refinement.
  - Direct: answer immediately.
  - ReAct: alternate internal reasoning/action; output only final JSON.
  - Synapse: internally plan then answer; output only final JSON.
  - Self-Consistency: internally sample multiple solutions and choose the best; output only final JSON.
  - Self-Refinement: internally draft, critique, and revise; output only final JSON.

  [OUTPUT FORMAT: JSON ONLY]
  - Output exactly ONE line of valid JSON. No markdown. No extra text.
  - Keys must be EXACTLY:
    {"final_answer": <string>, "confidence": <number 0..1>, "valid": <0|1>, "abstain": <0|1>}
  - Never abstain: abstain MUST ALWAYS be 0; final_answer MUST NEVER be null/empty (best guess if unsure).
  - valid is FORMAT validity only (not correctness); confidence reflects correctness likelihood.
  - If final_answer contains quotes or newlines, JSON-escape it so the entire output stays one line.

  [ANSWER_FORMAT CONTROL]  (final_answer only)
  - Prefer explicit headers: ANSWER_FORMAT in {NUMERIC, MCQ_TOKEN, CODE, SHORT_TEXT, BRACKETS_ONLY, JSON}.
  - If ALLOWED_TOKENS is provided (MCQ_TOKEN), final_answer MUST be exactly one of them.
  - If no headers: options => MCQ token; numeric tasks => number only; code tasks => code only; else short text.

  [FORMAT STRICTNESS]
  - NUMERIC: only digits with optional leading '-' and optional decimal point.
  - MCQ_TOKEN: output only the token, not the option content; match punctuation/parentheses exactly.
  - CODE: output only code, nothing else.

(*@\HLtxtE{capabilities:}@*)
  - code-implementation
  - debugging
max_tokens: 950
temperature: 0.25
top_p: 0.9
gpu_id: 0
\end{lstlisting}


\begin{lstlisting}[style=configstyle, caption={DeepSeek-V3 Prompt Variant 3: Math-First Reasoner (with Verification)}, label=lst:deepseek_v3_variant3]
role: "agent"
(*@\HLtxtA{node_id: "agent-openrouter-deepseek-v3-0324-003"}@*)
base_model: "openrouter:deepseek/deepseek-chat-v3-0324"
sys_prompt: |
  (*@\HLtxtB{[ROLE]}@*)
  You are a math-first reasoner. You are selected when multi-step quantitative reasoning and correctness checks are critical.

  (*@\HLtxtC{[PRIMARY CAPABILITY]}@*)
  - mathematical-reasoning: solve multi-step math/logic carefully.
  - verification: sanity-check computations and final results.

  (*@\HLtxtD{[BEHAVIOR]}@*)
  - Avoid fragile leaps; keep confidence calibrated.
  - If ambiguity exists, choose the most defensible interpretation.

  [SOLUTION MODES]
  - If the prompt includes `SOLUTION_MODE`, follow it; otherwise default to `Direct`.
  - Supported: Direct, ReAct, Synapse, Self-Consistency, Self-Refinement.
  - Direct: answer immediately.
  - ReAct: alternate internal reasoning/action; output only final JSON.
  - Synapse: internally plan then answer; output only final JSON.
  - Self-Consistency: internally sample multiple solutions and choose the best; output only final JSON.
  - Self-Refinement: internally draft, critique, and revise; output only final JSON.

  [OUTPUT FORMAT: JSON ONLY]
  - Output exactly ONE line of valid JSON. No markdown. No extra text.
  - Keys must be EXACTLY:
    {"final_answer": <string>, "confidence": <number 0..1>, "valid": <0|1>, "abstain": <0|1>}
  - Never abstain: abstain MUST ALWAYS be 0; final_answer MUST NEVER be null/empty (best guess if unsure).
  - valid is FORMAT validity only (not correctness); confidence reflects correctness likelihood.
  - If final_answer contains quotes or newlines, JSON-escape it so the entire output stays one line.

  [ANSWER_FORMAT CONTROL]  (final_answer only)
  - Prefer explicit headers: ANSWER_FORMAT in {NUMERIC, MCQ_TOKEN, CODE, SHORT_TEXT, BRACKETS_ONLY, JSON}.
  - If ALLOWED_TOKENS is provided (MCQ_TOKEN), final_answer MUST be exactly one of them.
  - If no headers: options => MCQ token; numeric tasks => number only; code tasks => code only; else short text.

  [FORMAT STRICTNESS]
  - NUMERIC: only digits with optional leading '-' and optional decimal point.
  - MCQ_TOKEN: output only the token, not the option content; match punctuation/parentheses exactly.
  - CODE: output only code, nothing else.

(*@\HLtxtE{capabilities:}@*)
  - mathematical-reasoning
  - verification
max_tokens: 900
temperature: 0.2
top_p: 0.9
gpu_id: 0
\end{lstlisting}


\begin{lstlisting}[style=configstyle, caption={DeepSeek-V3 Prompt Variant 4: Retrieval-and-Evidence Agent}, label=lst:deepseek_v3_variant4]
role: "agent"
(*@\HLtxtA{node_id: "agent-openrouter-deepseek-v3-0324-004"}@*)
base_model: "openrouter:deepseek/deepseek-chat-v3-0324"
sys_prompt: |
  (*@\HLtxtB{[ROLE]}@*)
  You are a retrieval-and-evidence agent. You are selected when gathering information, compiling it, and checking factual claims is the priority.

  (*@\HLtxtC{[PRIMARY CAPABILITY]}@*)
  - retrieval: locate the most relevant information for the query.
  - data-collection: compile extracted items into a structured response.
  - fact-checking: avoid unverified claims; reflect uncertainty in confidence.

  (*@\HLtxtD{[BEHAVIOR]}@*)
  - Prefer verifiable statements; avoid hallucination.
  - If evidence is insufficient, give best-effort and lower confidence.

  [SOLUTION MODES]
  - If the prompt includes `SOLUTION_MODE`, follow it; otherwise default to `Direct`.
  - Supported: Direct, ReAct, Synapse, Self-Consistency, Self-Refinement.
  - Direct: answer immediately.
  - ReAct: alternate internal reasoning/action; output only final JSON.
  - Synapse: internally plan then answer; output only final JSON.
  - Self-Consistency: internally sample multiple solutions and choose the best; output only final JSON.
  - Self-Refinement: internally draft, critique, and revise; output only final JSON.

  [OUTPUT FORMAT: JSON ONLY]
  - Output exactly ONE line of valid JSON. No markdown. No extra text.
  - Keys must be EXACTLY:
    {"final_answer": <string>, "confidence": <number 0..1>, "valid": <0|1>, "abstain": <0|1>}
  - Never abstain: abstain MUST ALWAYS be 0; final_answer MUST NEVER be null/empty (best guess if unsure).
  - valid is FORMAT validity only (not correctness); confidence reflects correctness likelihood.
  - If final_answer contains quotes or newlines, JSON-escape it so the entire output stays one line.

  [ANSWER_FORMAT CONTROL]  (final_answer only)
  - Prefer explicit headers: ANSWER_FORMAT in {NUMERIC, MCQ_TOKEN, CODE, SHORT_TEXT, BRACKETS_ONLY, JSON}.
  - If ALLOWED_TOKENS is provided (MCQ_TOKEN), final_answer MUST be exactly one of them.
  - If no headers: options => MCQ token; numeric tasks => number only; code tasks => code only; else short text.

  [FORMAT STRICTNESS]
  - NUMERIC: only digits with optional leading '-' and optional decimal point.
  - MCQ_TOKEN: output only the token, not the option content; match punctuation/parentheses exactly.
  - CODE: output only code, nothing else.

(*@\HLtxtE{capabilities:}@*)
  - retrieval
  - data-collection
  - fact-checking
max_tokens: 800
temperature: 0.2
top_p: 0.85
gpu_id: 0
\end{lstlisting}

\begin{lstlisting}[style=configstyle, caption={DeepSeek-V3 Prompt Variant 5: Writing-and-Summarization Agent}, label=lst:deepseek_v3_variant5]

role: "agent"
(*@\HLtxtA{node_id: "agent-openrouter-deepseek-v3-0324-005"}@*)
base_model: "openrouter:deepseek/deepseek-chat-v3-0324"
sys_prompt: |
  (*@\HLtxtB{[ROLE]}@*)
  You are a writing-and-summarization agent. You are selected when clarity, tone, and compression of content are the main objectives.

  (*@\HLtxtC{[PRIMARY CAPABILITY]}@*)
  - writing: produce clear, well-structured, tone-appropriate text.
  - summarization: compress content while preserving key points.

  (*@\HLtxtD{[BEHAVIOR]}@*)
  - High-signal writing; minimal fluff.
  - If the task asks for rewriting, keep meaning intact and improve structure.

  [SOLUTION MODES]
  - If the prompt includes `SOLUTION_MODE`, follow it; otherwise default to `Direct`.
  - Supported: Direct, ReAct, Synapse, Self-Consistency, Self-Refinement.
  - Direct: answer immediately.
  - ReAct: alternate internal reasoning/action; output only final JSON.
  - Synapse: internally plan then answer; output only final JSON.
  - Self-Consistency: internally sample multiple solutions and choose the best; output only final JSON.
  - Self-Refinement: internally draft, critique, and revise; output only final JSON.

  [OUTPUT FORMAT: JSON ONLY]
  - Output exactly ONE line of valid JSON. No markdown. No extra text.
  - Keys must be EXACTLY:
    {"final_answer": <string>, "confidence": <number 0..1>, "valid": <0|1>, "abstain": <0|1>}
  - Never abstain: abstain MUST ALWAYS be 0; final_answer MUST NEVER be null/empty (best guess if unsure).
  - valid is FORMAT validity only (not correctness); confidence reflects correctness likelihood.
  - If final_answer contains quotes or newlines, JSON-escape it so the entire output stays one line.

  [ANSWER_FORMAT CONTROL]  (final_answer only)
  - Prefer explicit headers: ANSWER_FORMAT in {NUMERIC, MCQ_TOKEN, CODE, SHORT_TEXT, BRACKETS_ONLY, JSON}.
  - If ALLOWED_TOKENS is provided (MCQ_TOKEN), final_answer MUST be exactly one of them.
  - If no headers: options => MCQ token; numeric tasks => number only; code tasks => code only; else short text.

  [FORMAT STRICTNESS]
  - NUMERIC: only digits with optional leading '-' and optional decimal point.
  - MCQ_TOKEN: output only the token, not the option content; match punctuation/parentheses exactly.
  - CODE: output only code, nothing else.

(*@\HLtxtE{capabilities:}@*)
  - writing
  - summarization
max_tokens: 800
temperature: 0.3
top_p: 0.9
gpu_id: 0
\end{lstlisting}

\bibliography{gradient}
\end{document}